%% file: faster_edge_coloring.tex
\documentclass[11pt]{article}
\usepackage{amsmath,amssymb,amsfonts,latexsym,graphicx,amsthm}
\usepackage{fullpage,color}
\usepackage{url,hyperref}
\usepackage{comment}
\usepackage[linesnumbered,boxed,ruled,vlined]{algorithm2e}
\usepackage{framed}
\usepackage[shortlabels]{enumitem}
\usepackage{cleveref}
\usepackage{tikz}
\usetikzlibrary{calc}
\usetikzlibrary{intersections}

\usepackage[framemethod=tikz]{mdframed}
\usepackage{tikz-cd}
\newenvironment{wrapper}[1]
{
	\begin{center}
		\begin{minipage}{\linewidth}
			\begin{mdframed}[hidealllines=true, backgroundcolor=gray!20, leftmargin=0cm,innerleftmargin=0.4cm,innerrightmargin=0.4cm,innertopmargin=0.4cm,innerbottommargin=0.4cm,roundcorner=10pt]
				#1}
			{\end{mdframed}
		\end{minipage}
	\end{center}
}

\definecolor{ForestGreen}{rgb}{0.1333,0.5451,0.1333}
\definecolor{DarkRed}{rgb}{0.65,0,0}

\hypersetup{
    colorlinks=true,
    linkcolor=DarkRed,
    filecolor=magenta,      
    urlcolor=cyan,
    citecolor=ForestGreen,
}

\usepackage[margin=1in]{geometry}

\newtheorem{theorem}{Theorem}[section]
\newtheorem{lemma}{Lemma}[section]
\newtheorem{definition}{Definition}[section]

\newtheorem{claim}{Claim}[section]

\newtheorem{invariant}{Invariant}[section]

        {\medskip}

\newcommand{\ceil}[1]{\left\lceil #1 \right\rceil}
\newcommand{\floor}[1]{\left\lfloor #1 \right\rfloor}

\newcommand{\brac}[1]{\left(#1\right)}
\newcommand{\miss}{\mathsf{miss}}
\newcommand{\clr}{\mathsf{clr}}
\newcommand{\lst}{\mathsf{lst}}
\newcommand{\ds}{\mathcal{D}}
\newcommand{\br}{\mathsf{br}}
\newcommand{\paths}{\mathcal{P}}
\newcommand{\lone}{\mathsf{lone}}
\newcommand{\desc}{\mathsf{desc}}
\newcommand{\cng}{\mathsf{cng}}
\newcommand{\hi}{\mathrm{hi}}
\newcommand{\lo}{\mathrm{lo}}

\begin{document}

\title{Faster $(\Delta + 1)$-Edge Coloring: Breaking the $m \sqrt{n}$ Time Barrier}
\author{
	Sayan Bhattacharya$^*$ \and Din Carmon\textsuperscript{\textdagger} \and Mart\'in Costa$^*$ \and Shay Solomon\textsuperscript{\textdagger} \and Tianyi Zhang\textsuperscript{\textdagger}
}

\date{Warwick University$^*$ \\ Tel Aviv University\textsuperscript{\textdagger}}
\maketitle

\begin{abstract}
Vizing's theorem states that 
any $n$-vertex $m$-edge graph of maximum degree $\Delta$ can be {\em edge colored} using at most $\Delta + 1$ different colors [Diskret.~Analiz, '64].
    Vizing's original proof is algorithmic and shows that such an edge coloring can be found in $\tilde{O}(mn)$ time. This was subsequently improved to $\tilde O(m\sqrt{n})$, independently by Arjomandi [1982] and by Gabow et al.~[1985]. 
    
    In this paper we present an algorithm that computes such an edge coloring in $\tilde O(mn^{1/3})$ time, giving the first polynomial improvement for this fundamental problem in over 40 years.
\end{abstract}

\thispagestyle{empty}
\clearpage
\setcounter{page}{1}

\input{intro}

\input{sayan-overview}

\input{prelim}
\input{alg}

\vspace{5mm}
\bibliographystyle{alpha}
\bibliography{ref}

\end{document}

%% file: intro.tex
\section{Introduction}
\label{sec:intro}

Consider a simple, undirected graph $G = (V, E)$ with maximum degree $\Delta$, and an integer $\kappa \in \mathbb{N}^+$. A $\kappa$-edge coloring $\psi : E \rightarrow \{1, 2, \ldots, \kappa\}$ of $G$ assigns a ``color'' $\psi(e)$ to every edge $e \in E$, in such a way that no two adjacent edges receive the same color. What is the minimum  $\kappa$ for which $G$ admits a $\kappa$-edge coloring? This value is known as the ``edge chromatic number'' of $G$. Any edge coloring trivially requires at least $\Delta$ colors. On the other hand, a textbook theorem by Vizing guarantees that at most $\Delta+1$ colors are always  sufficient~\cite{vizing1965critical}, and it is NP-complete to distinguish whether the edge chromatic number of a given  graph is  $\Delta$ or $\Delta+1$~\cite{holyer1981np}. This leads  naturally to a fundamental question in graph algorithms, summarized below.

\begin{wrapper}
How fast can we compute a $(\Delta+1)$-edge coloring of an input graph with $n$ vertices and $m$ edges?
\end{wrapper}

\noindent
{\bf State-of-the-Art.} Vizing's original proof can easily be converted into  an $O(m n)$ time algorithm for $(\Delta+1)$-edge coloring. Back in the 1980s, Arjomandi~\cite{arjomandi1982efficient} and Gabow et al.~\cite{gabow1985algorithms} independently improved this runtime bound to $\tilde{O}(m\sqrt{n})$.\footnote{Throughout the paper, the notation $\tilde{O}(.)$ hides polylogarithmic in $n$ factors.} Recently, Sinnamon~\cite{sinnamon2019fast}  simplified these previous algorithms of~\cite{arjomandi1982efficient,gabow1985algorithms} using randomization, alongside shaving off some logarithmic factors from their runtimes, to achieve a clean bound of $O(m \sqrt{n})$. {\em But,  there has been  no polynomial improvement over this $m\sqrt{n}$ time barrier in over four decades}. We breach this barrier with the following result.

\begin{theorem}
\label{faster-vizing}
Given a simple undirected graph $G = (V, E)$ on $n$ vertices and $m$ edges with maximum degree $\Delta$, there is a randomized algorithm that computes a $(\Delta+1)$-edge coloring with high probability in runtime $\tilde{O}\brac{mn^{\frac{1}{3}}}$.
\end{theorem}

\subsection{Related Work}

 If we allow for a larger palette of $\Delta+\tilde{O}(\sqrt{\Delta})$ colors, then the problem can be solved in $\tilde{O}(m)$ time~\cite{karloff1987efficient}. In addition, there exist  algorithms which  achieve linear or near-linear runtime for $(1+\epsilon)\Delta$-edge coloring~\cite{duan2019dynamic,BhattacharyaCPS24,elkin2024deterministic}.

 Quite a few improved  results have been obtained over the years
 for restricted graph classes. 
 For bipartite graphs one can always compute a $\Delta$-edge coloring in $\tilde{O}(m)$ time~\cite{combinatorica/ColeOS01}. For bounded degree graphs one can  compute a $(\Delta+1)$-edge coloring in $\tilde{O}(m \Delta)$ time \cite{gabow1985algorithms}, a classic result that was generalized recently for
  bounded arboricity graphs~\cite{BhattacharyaCPS24b}; see \cite{BhattacharyaCPS24c,ChristiansenRV24,Kowalik24} for further recent works on edge coloring in bounded arboricity graphs.
 Subfamilies of bounded arboricity graphs, in particular planar graphs, bounded treewidth graphs and bounded genus graphs, were studied in \cite{chrobak1989fast,chrobak1990improved,cole2008new}.

 Finally, in recent years  substantive  effort has been devoted to the study of the edge coloring problem in other computational models; such as dynamic algorithms~\cite{BarenboimM17,BhattacharyaCHN18,duan2019dynamic,Christiansen23,BhattacharyaCPS24,Christiansen24}, online algorithms~\cite{CohenPW19,BhattacharyaGW21,SaberiW21,KulkarniLSST22,BilkstadSVW24}, distributed algorithms~\cite{panconesi2001some,elkin20142delta,fischer2017deterministic,ghaffari2018deterministic,balliu2022distributed,ChangHLPU20,Bernshteyn22,Christiansen23,Davies23}, and streaming algorithms~\cite{BehnezhadDHKS19,behnezhad2023streaming,chechik2023streaming,ghosh2023low}.

%% file: sayan-overview.tex
\section{Technical Overview}
\label{sec:overview:sayan}

    In this section, we will heavily use ideas from Vizing's original proof, which uses concepts like ``fans'' and ``alternating paths'' (see Section~\ref{prelim} for precise definitions of these concepts).

\subsection{Reviewing the State-of-the-Art}
\label{sub:sec:review}
We begin by outlining the  approach of~\cite{arjomandi1982efficient,gabow1985algorithms,sinnamon2019fast} that achieves the  runtime bound of $\tilde{O}(m \sqrt{n})$. This approach~\cite{arjomandi1982efficient,gabow1985algorithms,sinnamon2019fast} fits within the following abstract template. It {\em iterates} through the edges of $G = (V, E)$ in a carefully chosen order $e_1, \ldots, e_m$. Consider any $i \in [m]$. At the start of the $i^{th}$ iteration, the algorithm already has a valid partial coloring $\psi : \{e_1, \ldots, e_{i-1}\} \rightarrow [\Delta+1]$ of the edges seen until now. The goal of the current iteration is to perform a ``color-extension'' step, which {\em extends} the partial coloring $\psi$ to the edge-set $\{e_1, \ldots, e_{i-1}, e_i\}$, by assigning a color to $e_i$ and  changing the colors of some subset $C_{i} \subseteq \{e_1, \ldots, e_{i-1}\}$ of the previous edges. Define $|C_i|$ to be the ``cost'' of this color-extension step. At the end of the $m^{th}$ iteration, the algorithm returns the coloring $\psi$ of the entire graph $G = (V, E)$. Define $\sum_{i=1}^m |C_i|$ to be the ``cost of the algorithm'' on the given input. 

We now highlight a few crucial points. (i) During a given iteration, the color-extension step is performed just by following Vizing's recipe, which involves ``rotating a fan'' and ``flipping an alternating path''. There is nothing special about it. (ii) Instead, the key driver of the algorithm is the way it {\em chooses} the ordering $e_1, \ldots, e_m$ of the edges. (iii) The algorithm's cost is clearly a lower bound on its update time, because it needs to spend $\Omega(1)$ time whenever it changes the color of an edge; to be more precise, this property holds for all the known algorithms, as they 
%spend $\Omega(1)$ time on each edge in the alternating path
%since they 
recolor each edge {\em separately}. On the other hand, as long as the algorithm is not doing something {\em out of the ordinary}, it is typically possible to come up with supporting data structures so as to ensure that the algorithm's runtime is proportional to its cost (up to polylogarithmic factors). Indeed, the costs of all the existing state-of-the-art algorithms~\cite{arjomandi1982efficient,gabow1985algorithms,sinnamon2019fast} also happen to be $\Omega(m\sqrt{n})$. 

\medskip
\noindent {\bf Two Simplifying Assumptions.} In light of the above discussion, {\em for  now} we make two assumptions, which will help us emphasize the key technical insight in this paper. 
\begin{itemize}
\item (a) We focus only on designing an algorithm {\bf whose {\em cost} breaches the $m\sqrt{n}$ barrier}. 
\item (b)  While performing a color-extension step, we pretend that {\bf we only need to flip an alternating path starting from one of the endpoints of the uncolored edge} without rotating any Vizing fan. To be a bit more precise, we assume  that whenever we are asked to extend the current partial coloring to an uncolored edge $(u, v)$, we get lucky in the following sense: the relevant alternating path $P$ starting from $v$ does not end  at the other endpoint $u$ (which always happens in bipartite graphs, for example). So, it suffices to only flip the colors on the path $P$, without us having to rotate any Vizing fan.
\end{itemize}

The previous algorithms~\cite{arjomandi1982efficient,gabow1985algorithms,sinnamon2019fast} rely on two ingredients, as outlined below.  Actually, \cite{sinnamon2019fast} uses this ingredient explicitly, whereas the deterministic algorithms from \cite{arjomandi1982efficient,gabow1985algorithms} use this idea implicitly in a more intricate way.

\medskip
\noindent {\bf Ingredient I (Random Sampling an Uncolored Edge).} This ingredient helps us achieve the $\tilde{O}(m\sqrt{n})$ bound when $\Delta \leq \sqrt{n}$. To see how it works, suppose that we have a partial coloring $\psi$ in $G = (V, E)$, which leaves us with a set $E' \subseteq E$ of currently uncolored edges. If we wish to extend this coloring to any  edge $e' \in E'$ using Vizing's recipe, then we have to flip an alternating path $P_{e'}$ (starting from some endpoint of $e'$). Let $L_{\texttt{total}}$ denote the sum of the lengths of these alternating paths, over all $e' \in E'$. In the following paragraph, we will show that $L_{\texttt{total}} \leq \Delta m$.

Fix an already colored edge $(u, v) \in E \setminus E'$, and let $c = \psi(u, v)$. Note that for any color $c' \in [\Delta+1] \setminus \{c\}$, the edge $(u, v)$ can appear in at most one $\{c, c'\}$-alternating path in $\{ P_{e'} \}_{e' \in E'}$. This holds because the subgraph of $G$ consisting of the edges $e$ with colors $\psi(e) \in \{c, c'\}$ is a collection of vertex-disjoint paths and cycles. Thus, summing over all such $c' \in [\Delta+1] \setminus \{c\}$, we infer that the edge $(u, v)$ appears in  at most $\Delta$  alternating paths from the collection $\{ P_{e'} \}_{e' \in E'}$. Since there are at most $m$ such edges  $(u, v) \in E \setminus E'$, a simple counting argument gives us:
\begin{equation}
\label{eq:bound:length:1}
L_{\texttt{total}}\leq \Delta m.
\end{equation}

Accordingly, if we now sample an uncolored edge $e' \in E'$, then in expectation the length of the concerned alternating path $P_{e'}$ will be at most $\Delta m/|E'|$,  and this will be an upper bound on our expected cost for this color-extension step. 

We now derive the following natural algorithm: at each iteration, pick a currently uncolored edge $e'$ uniformly at random, and extend the coloring to $e'$. From the discussion above, the total expected cost of this algorithm is at most: $\sum_{\lambda=m}^1 (\Delta m)/\lambda = \tilde{O}(\Delta m)$, where $\lambda$ refers to the number of uncolored edges at the start of a concerned iteration. Finally, note that $\Delta m \leq m \sqrt{n}$ if $\Delta \leq \sqrt{n}$.

\medskip
\noindent {\bf Ingredient II (Divide, Conquer and Combine).} This ingredient helps us achieve the $\tilde{O}(m \sqrt{n})$ bound when $\Delta \geq \sqrt{n}$. Here, we crucially exploit a known subroutine (see Lemma~\ref{euler}). In $\tilde{O}(m)$ time, this allows us to partition the input graph $G = (V, E)$ into two (almost equal sized) subgraphs $G_1 = (V, E_1)$ and $G_2 = (V, E_2)$, such that each of $G_1$ and $G_2$ has maximum degree at most $\lceil \Delta/2 \rceil$. We then recursively color $G_1$ and $G_2$, using two {\em mutually disjoint} palettes of $\lceil \Delta/2\rceil+1$ colors. Putting them together, this gives us a valid edge coloring of the whole graph $G = G_1 \cup G_2$, using $2 \cdot (\lceil \Delta/2 \rceil+1) \leq \Delta + 3$ colors. Next, we identify the two ``least popular'' colors in $G$ (these are the colors that get assigned to the least number of edges). Let $c_0$ and $c_1$ be these two least popular colors, and let $E' := \{ e \in E : \psi(e) \in \{c_0, c_1\}\}$ denote the concerned set of edges that receive them. A simple averaging argument implies that $|E'| = O(m/\Delta)$. We now uncolor the edges in $E'$. This gives us a valid {\em partial} $(\Delta+1)$-edge coloring $\psi$ of $G = (V, E)$, with a set $E' \subseteq E$ of $O(m/\Delta)$ uncolored edges. We refer to the set $E'$ as ``leftover edges''. To wrap things up, we extend the $(\Delta+1)$-edge coloring $\psi$ to these leftover edges, by going through them in any arbitrary order and coloring each of them according to Vizing's recipe. Trivially, each of these color-extensions for the leftover edges incurs a cost of $n$ (which is an upper bound on  the length of the concerned alternating path that we need to flip). The ``combine'' step of this divide, conquer and combine approach, therefore, incurs a total cost of $O(mn/\Delta)$. Let $T(m, \Delta)$ denote the cost of the overall algorithm. We accordingly get the recurrence stated below:

\begin{equation}
\label{eq:recur:1}
T(m, \Delta) = 2 \cdot T\left(\ceil{\frac{m}{2}}, \ceil{\frac{\Delta}{2}}\right) + O\left(\frac{mn}{\Delta}\right). 
\end{equation}
Note that if $\Delta \geq \sqrt{n}$, then the second term in the RHS of~(\ref{eq:recur:1}) becomes $O(m \sqrt{n})$. Thus, we can unfold the above recurrence, truncating the ``recursion tree'' at the ``level'' where $\Delta$ becomes equal to $\sqrt{n}$, and apply Ingredient I to solve the resulting subproblems  at that last level (where $\Delta \simeq \sqrt{n}$). It is easy to verify that this framework also gives us  the   bound of $\tilde{O}(m\sqrt{n})$ for the overall cost.

\subsection{A Canonical Instance (and How We Solve It)}
\label{sub:sec:instance}

For the rest of Section~\ref{sec:overview:sayan}, we focus on the following input instance: $G = (V, E)$ is a {\em regular} graph with degree $\Delta = \sqrt{n}$, and so  $m = \Theta(n \Delta) = \Theta(n \sqrt{n})$ . Since $\Delta = \sqrt{n}$ and the graph is $\sqrt{n}$-regular, here the existing algorithms discussed in Section~\ref{sub:sec:review} will achieve a bound of $\tilde{\Theta}(m \sqrt{n}) = \tilde{\Theta}(n^2)$. 
\textbf{Can we beat the $\tilde{O}(n^2)$ bound on this instance?} This is the question we will attempt to address below. This canonical instance is in some sense the hard instance: first, adding the regularity restriction does not lose much generality, and second this choice of parameters is the ``balance point'' of the two ingredients described earlier, i.e., where both meet. Thus, once we have resolved it, the same approach can be applied, with some extra work, to resolve the general case.

\subsubsection{Our Algorithmic Template} 
\label{sub:sec:template}
We sample each vertex $v \in V$ u.a.r.~with probability $\simeq n^{-1/4}$ into a set $U$.\footnote{Actually, the sampling probability would be $(c \log n) \cdot n^{-1/4}$, so that we are able to apply standard concentration bounds. We ignore these extra polylogarithmic factors to highlight the main idea.}   Using standard concentration bounds, w.h.p.\ we  have: (i) $|U| \simeq n^{3/4}$, (ii) $\Delta(G [V \setminus U]) \simeq n^{1/2} - n^{1/4}$ and (iii) $\Delta(G[U]) \simeq n^{1/4}$; where $\Delta(G[S])$ is the maximum degree in the subgraph of $G$ {\em induced} by $S \subseteq V$.

We next separately compute  $(\Delta+1)$-edge colorings  of $G[V \setminus U]$ and $G[U]$, using colors from the same palette $[\Delta+1]$. Because of  the  $n^{1/4} = \sqrt{\Delta}$ ``slack'' guaranteed by property (ii), it takes only $\tilde{O}(m)$ time to color  $G[V \setminus U]$  by invoking the algorithm of~\cite{karloff1987efficient}.\footnote{The algorithm by~\cite{karloff1987efficient} computes a $(\Delta+\tilde{O}(\sqrt{\Delta}))$-edge coloring in near-linear time.} In contrast, we  color $G[U]$ using a greedy algorithm in $O(m)$ time. We can do this because $\Delta(G[U])$ is negligible compared to $\Delta+1$, and the vertices in $G[U]$ are disjoint from the vertices in $G[V \setminus U]$. To summarize, at this point all but the edges in $G[U \times (V \setminus U)]$ are colored, and the cost we have incurred until now is only $\tilde{O}(m)$. We think of the edges in $G[U \times (V \setminus U)]$ as a collection of ``stars'', with the vertices in $U$ as  ``centers'' and the vertices in $V \setminus U$ as ``clients''. It remains to show how to color these star edges.

\medskip
\noindent {\bf Coloring the stars.}  For all $v \in V$, let $\miss_{\psi}(v) \subseteq [\Delta+1]$ denote the set of ``missing colors'' at $v$, under the current partial coloring $\psi$. Specifically, these are the colors that are {\em not}  assigned to any edge $e \in E$ incident on  $v$. Throughout the algorithm, for every star edge $(u, v) \in G[U \times (V \setminus U)]$, the client $v \in V \setminus U$ will maintain a ``tentative color'' $\clr(v \rightarrow u) \in \miss_{\psi}(v)$ for the center $u \in U$, while ensuring that Invariant~\ref{inv:tentative} holds. It is easy to satisfy this invariant because the palette  consists of $\Delta+1$ colors, whereas any client $v \in V \setminus U$ is incident upon at most $\Delta$ edges in $G = (V, E)$.

\begin{invariant}
\label{inv:tentative}
For every two distinct uncolored edges $(u, v), (u', v) \in G[U \times (V \setminus U)]$ sharing a common client $v \in V \setminus U$, we have $\clr(v \rightarrow u) \neq \clr(v \rightarrow u')$.
\end{invariant}

As usual, we will color the star edges in iterations. Each iteration will identify an appropriate uncolored star edge $e$, and extend the  $(\Delta+1)$-coloring $\psi$ to $e$. We terminate the algorithm when there are no more uncolored edges left. Below, we explain how a given iteration is implemented. 

Pick a color $x \in [\Delta+1]$ u.a.r. Let $X := \{ u \in U : x \in \miss_{\psi}(u)\}$ be the collection of centers that have $x$ as a missing color, and let $H$ denote the set of uncolored star edges incident on the centers in $X$. Next, u.a.r.~pick an  edge $(u', v') \in H$. W.l.o.g., suppose that $u'$ is a center and $v'$ is a client. Let $P_x(v' \rightarrow u')$ denote the unique $\{x, \clr(v' \rightarrow u')\}$-alternating path in $G$ starting from $v'$. Flip the colors on  $P_{x}(v' \to u')$, and assign the color $x$ to the edge $(u', v')$. This completes the color-extension step. The cost incurred during this step equals the length of the path $P_{x}(v' \to u')$. 

%We next derive the following key lemma.

\begin{lemma}
\label{lm:overview:1} Let $m_0$ denote the number of uncolored star edges at the start of an iteration. Then the expected cost incurred during this iteration is at most $O\left(\min\left(n, \frac{\Delta^2 n |U|}{m_0^2}\right)\right) = O\left(\min\left(n, \frac{n^{2+3/4}}{m_0^2} \right)\right)$.
\end{lemma}

\medskip
\noindent {\bf Bounding the total cost.} Before outlining the proof sketch of Lemma~\ref{lm:overview:1}, we explain how this lemma implies that the total cost of our algorithm is polynomially better than the  $\tilde{O}(n^2)$ threshold. Consider the sequence of iterations in which we extend the coloring to the star edges, one after the other. Partition this sequence into two phases. Phase I starts in the beginning, when $m_0$ (the number of uncolored edges) equals $|U| \cdot \Delta$, and ends when $m_0$  becomes equal to some threshold $m^{\star} = n^{4/5}$. Phase II consists of the remaining $m^{\star}$ iterations. By Lemma~\ref{lm:overview:1}, during an iteration in Phase I, the expected cost incurred is at most: $$O\left(\frac{n^{2+3/4}}{m_0^2} \right) \leq O  \left(\frac{n^{2+3/4}}{m_0 \cdot m^{\star}} \right) = O\left(\frac{n^{2+3/4-4/5}}{m_0} \right) = O\left( \frac{n^{19/20}}{m_0}\right).$$
Next, the cost incurred during any iteration in Phase II is at most $n$ (this also follows from Lemma~\ref{lm:overview:1}). Finally, we have already discussed that we pay a cost of only $\tilde{O}(m) = \tilde{O}(n^{3/2})$ for coloring the edges in $G$ that are {\em not} part of the stars. Thus, the total cost we pay to color the whole input graph $G$ is at most:
$$\sum_{m_0 =m^{\star}}^{|U| \cdot \Delta} O\left( \frac{n^{19/20}}{m_0}\right) + m^{\star} \cdot n + \tilde{O}(n^{3/2}) = \tilde{O}\left(n^{19/20}\right)+n^{9/5} + \tilde{O}\left(n^{3/2}\right) = \tilde{O}\left(n^{19/20} \right).$$
{\bf Remark.} By balancing out the costs incurred in different parts of the algorithm against one another, we can further improve the bound derived above. The main focus of this section, however, is to show how to get a polynomial improvement over the state-of-the-art bound of $\tilde{O}(n^2)$ on this specific instance, which stood for nearly four decades. We have already succeeded in this goal.

\medskip
\noindent {\bf Proof sketch of Lemma~\ref{lm:overview:1}.}
Let $C$ denote the expected cost incurred during the concerned iteration. Since the length of any alternating path in $G$ is at most $n$, we trivially have $C \leq n$. It remains to show that $C \leq O(\Delta^2 n |U|/m_0^2)$. Towards this end, we will start by lower bounding the expected size of $H$.

For each center $u \in U$, let $d_u$ denote its ``uncolored degree'', i.e., the number of uncolored star edges incident on $u$. It is easy to observe that $|\miss_{\psi}(u)| \geq d_u+1$, because the palette consists of $\Delta+1$ colors but the overall degree of $u$ is at most $\Delta$. Since $x \in [\Delta+1]$ is sampled u.a.r., a given center $u \in U$ belongs to the set $X$ with probability at least $\frac{|\miss_{\psi}(u)|}{\Delta+1} \geq \frac{d_u+1}{\Delta+1} \geq \frac{d_u}{\Delta}$, and when it does, it contributes $d_u$ many incident edges to the set $H$. Thus, in expectation, at least $d_u \cdot (d_u/\Delta) = d_u^2/\Delta$ many edges incident on $u$ are present in the set $H$. Since $\sum_{u \in U} d_u = m_0$, summing this bound over all the centers  $u \in U$, we get:
\begin{equation}
\label{eq:overview:1}
\mathbb{E}[|H|] \geq \sum_{u \in U} \frac{d_u^2}{\Delta} \geq \frac{\left(\sum_{u \in U} d_u \right)^2}{|U| \cdot \Delta} = \frac{m_0^2}{|U| \cdot \Delta}.
\end{equation}
In the above derivation, the second-step follows from an application of Cauchy-Schwarz inequality.

Consider the collection of alternating paths $\mathcal{P}_x := \{ P_x(v' \to u') : (v', u') \in H, u' \in X\}$. Each edge $(u', v') \in H$ contributes a unique path to $\mathcal{P}_{x}$, {\bf which gives us a natural one to one mapping between $H$ and $\mathcal{P}_x$}. Thus, we  essentially sample a path $P \in \mathcal{P}_x$ u.a.r.~and flip the colors on it. So, our expected cost is equal to $|P|$, the  length of this path $P$ that we sample. Fix any color $c \in [\Delta+1]$, and let $\mathcal{P}_{x,c} := \{ P_{x}(v' \to u') \in \mathcal{P}_x : \clr(v' \to u') = c\}$ denote the subset of these paths where the client has tentative color $c$ for the concerned center. Because of Invariant~\ref{inv:tentative}, each client $v \in V \setminus U$ contributes at most one path to the collection $\mathcal{P}_{x,c}$. As the edges with colors $x$ and $c$ in $G$ form a collection of vertex-disjoint paths and cycles, and each such path is counted at most twice in the collection $\mathcal{P}_x$ (once for each of its endpoints that is a client), this implies that $\sum_{P \in \mathcal{P}_{x, c}} |P| = O(n)$. Next, summing over all the colors $c\in [\Delta+1]$, we get:
\begin{equation}
\label{eq:paths:1}
\sum_{P \in \mathcal{P}_x} |P| = \sum_{c \in [\Delta+1]} \sum_{P \in \mathcal{P}_{x,c}} |P| =  O(\Delta n).
\end{equation}
 From~(\ref{eq:overview:1}) and~(\ref{eq:paths:1}), we now derive that the expected length of the alternating path $P$ that we need to flip (which equals the expected cost we incur) is at most:\footnote{We are being a bit  sloppy in this very high level proof sketch, in assuming that $|H| = \Omega(m_0^2/(|U| \cdot \Delta))$, whereas~(\ref{eq:overview:1}) holds only in expectation}
\begin{equation}
\label{eq:map:1}
\frac{\sum_{P \in \mathcal{P}_x} |P|}{|\mathcal{P}_x|} = 
\frac{\sum_{P \in \mathcal{P}_x} |P|}{|H|} = O\left(\frac{\Delta n}{m_0^2/(|U| \cdot \Delta)}\right) = O\left(\frac{\Delta^2 n |U|}{m_0^2}  \right).
\end{equation}
This concludes the proof sketch.
\qed

\subsubsection{Looking Back: The Key Insight}
\label{sub:sec:insight}

In summary, the main insight which allows us  to achieve a polynomial improvement over the previous approach~\cite{arjomandi1982efficient,gabow1985algorithms,sinnamon2019fast} is this: We carefully ensure that the uncolored edges are all incident on relatively few ``centers'', so that together they form a collection of ``stars''. We then carefully exploit this very special structure of the uncolored edges, to argue that we can perform a color-extension on one of them by flipping a relatively short  alternating path.  

To elaborate on this in a bit more detail, consider the $\sqrt{n}$-regular graph $G = (V, E)$ defined at the start of Section~\ref{sub:sec:instance}. If we wish to solve this instance via the previous approach~\cite{arjomandi1982efficient,gabow1985algorithms,sinnamon2019fast}, then we will end up invoking Ingredient I (see Section~\ref{sub:sec:review}). This is because if we  invoke Ingredient II, then the recursion tree will get truncated at the very first level (as $\Delta = \sqrt{n}$) and it will revert back to  Ingredient I in any case. Thus, on this instance, the previous approach   samples an uncolored edge $e$ u.a.r.~and then extends the existing partial coloring to $e$. As discussed in Section~\ref{sub:sec:review}, this color-extension step incurs an expected cost of $O(\Delta m/\lambda)$, where $\lambda$ is the current number of uncolored edges. Note that here we have {\em not} used any {\em structural property of the subgraph induced by the uncolored edges}, at any point in the analysis. 

In contrast, suppose  we are somehow lucky, so that all the  $\lambda$ edges that are currently uncolored are incident on a ``center'' vertex $u$. Thus, the uncolored edges are of the form $(u, v_1), 
\ldots, (u,v_{\lambda})$. It is then easy to verify that $|\miss_{\psi}(u)| \geq \lambda+1$, where $\psi$ denotes the current partial coloring. Consider any $i \in [\lambda]$, and let $c_i \in \miss_{\psi}(v_i)$ be any arbitrary missing color at $v_i$. If we want to extend the coloring $\psi$ to the edge $(u, v_i)$, then we  have at least $(\lambda+1)$ many different alternating paths (starting from $v_i$)  to choose from, one for every pair of colors in  $\bigcup_{c \in \miss_{\psi}(u)}\{ c, c_i\}$. Summing over all $i \in [\lambda]$, this gives us at least $\Omega(\lambda^2)$ many possible alternating paths, for the star centered at $u$. Hence, using the same argument as in the analysis of Ingredient I (see Section~\ref{sub:sec:review}), we can now infer that a path sampled u.a.r.~from this collection will have an expected length of $O(\Delta m/\lambda^2)$. This gives us a ``polynomial advantage'', because if the uncolored edges were scattered arbitrarily around the input graph, then we will only get an upper bound of $O(\Delta m/\lambda)$ on the average length of an alternating path, as opposed to the upper bound of $O(\Delta m/\lambda^2)$ we are getting for a star.

The above discussion motivates our basic algorithmic template in Section~\ref{sub:sec:template}, where we sample relatively few centers $U$, and set things up in a manner so that we only need to worry about extending a partial coloring to the uncolored edges belonging to the stars around these centers. We believe  that this key take home message from our algorithm, that of exploiting the structure of the remaining uncolored edges, will find future applications in the edge coloring literature. 

\subsubsection{Dealing with the Fans: The Remaining (and Very Technical) Challenge}

Until now, we have been ignoring the aspect where we might need to rotate a Vizing fan while performing a color-extension (see Assumption (b) in Section~\ref{sub:sec:review}). Dealing with the fans create a very significant bottleneck, which we 
 need to overcome via an extremely delicate case analysis. 

Let's revisit the proof sketch of Lemma~\ref{lm:overview:1}. There,  we highlighted the existence of a one-to-one mapping between $H$ (the relevant set of uncolored edges) and $\mathcal{P}_x$ (the set of alternating paths that we  need to flip while extending the coloring to some edge in $H$). This was crucially used in the first equality in~\Cref{eq:map:1}, where we asserted that $|\mathcal{P}_x| = |H|$. While dealing with the Vizing fans, this property can completely break down, and we might end up in a situation where $|H| \gg |\mathcal{P}_x|$, thereby invalidating the derivation in~\Cref{eq:map:1}.

In a bit more detail, suppose that we wish to extend the existing partial  coloring to an uncolored edge $(u', v') \in H$, with $u' \in U$ being the center. In general, this involves (i) rotating a fan around  $u'$, and (ii) then flipping the colors on an alternating path $P''$ which starts from {\em some other client} $v'' \neq v'$ (in the same star centered around $u'$). Thus, the edge $(u', v')$ gets {\em mapped} to this path $P''$. The issue is that {\em many different edges} can now get mapped to the {\em same} alternating path. This, in turn,  leads to the problematic situation described in the preceding paragraph (i.e., $|H| \gg |\mathcal{P}_x|$).

To overcome this significant challenge, the execution of our algorithm will split into four steps, as highlighted next, based on the following four different cases.

\paragraph{Case I: Pairing the clients.} \label{overview-step1} To establish a lower bound on $|\paths_x|$, the first question is what would happen if we do not try to rotate Vizing fans at all, but only flip the alternating paths $P_{x}(v\rightarrow u)$. The upfront issue is that an alternating path $P_{x}(v\rightarrow u)$ could end at $u$, so flipping the colors of $P_x(v\rightarrow u)$ does not allow us to extend $\psi$ to $(u, v)$ immediately. To work around this issue, a basic observation is that if another client $v'\neq v$ is generating the same type $\{x, \clr(v\rightarrow u)\}$ of alternating path $P_{x}(v'\rightarrow u)$ (or equivalently, $\clr(v\rightarrow u) = \clr(v'\rightarrow u)$), then at most one of $P_{x}(v\rightarrow u), P_{x}(v'\rightarrow u)$ can end at $u$, which means at least one of the two color extensions would be valid; this is actually the same u-fan technique from \cite{gabow1985algorithms}. Step \ref{alg-step1} of our main algorithm will be dealing with this case.

Based on this important observation, we will mainly focus here on the case where clients $v$ are missing different colors so that the u-fan technique fails. In this case, we have to use Vizing's procedure to color the edge $(u, v)$. The main issue shall arise with Vizing's procedure, but now we have an extra condition that all the colors $\clr(v\rightarrow u)$ are different for distinct clients $v$ of $u$.
(In this high-level overview we make the simplifying assumption that either all clients can be paired together as described above, or all the colors $\clr(v\rightarrow u)$ are different for distinct clients $v$;
the actual algorithm and analysis are much more subtle.) 

\paragraph{Case II: Dealing with directed chains.} \label{overview-step2} As a wishful thought, since the clients $v$ are missing different colors, Vizing's procedure on edges $(u, v)$ might yield different alternating paths. However, this wishful thought is generally incorrect. Consider a case where $u$ has a sequence of neighbors $w_0, w_1, \ldots, w_l$ such that $\psi(u, w_j)\in \miss_{\psi}(w_{j-1})$, for all $1\leq j\leq l$, and $\psi(u, w_{l-1})\in \miss_\psi(w_l)$. If all the clients are missing some colors in $\{\psi(u, w_k)\mid 0\leq k <l\}$, then any edge $(u, v)$ would generate the same $\{x, \psi(u, w_{l-1})\}$-alternating path through Vizing's procedure. To bypass this issue, our key insight is that, if $v$ is missing some color $\psi(u, w_k)$ and another vertex $v'$ is missing some other color $\psi(u, w_{k'}), k'<k$, then flipping the $\{x, \psi(u, w_k)\}$-alternating path starting at $u$ can also lead to a valid color extension. See \Cref{overview-fan} for an illustration. Step \ref{alg-step2} of our main algorithm will be dealing with this case.

\begin{figure}
	\centering
	\input{figs/overview-fan}
	\caption{In this picture, A vertex $s$ is directed to $t$ if $\psi(u, t)\in \miss_\psi(s)$. Applying Vizing' procedure for either $(u, v)$ or $(u, v')$ would end up with the same alternating path starting with same $\{x, \psi(u, w_2)\}$-alternating path. We will show that the $\{x, \psi(u, w_1)\}$-alternating path can also be used for color extension.}\label{overview-fan}
\end{figure}
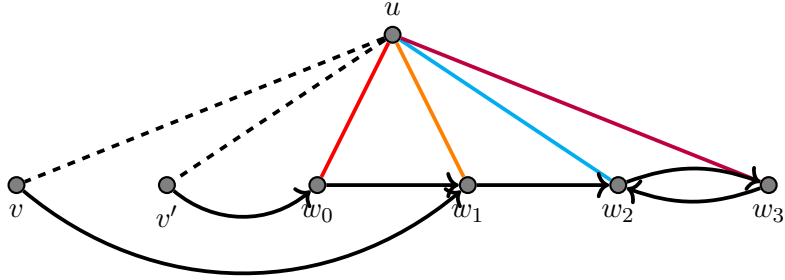

\paragraph{Case III: Dealing with directed stars.} \label{overview-step3} By the above insight, when the colored neighborhood of $u$ has a long directed chain structure, we are in good shape.  This technique fails miserably when the neighborhood forms a directed star structure instead of a chain; that is, all neighbors are directed toward the same center vertex, or more specifically, $\psi(u, w_{l})\in \miss_\psi(w_k)$, for all $k\neq l$, and $\psi(u, w_{l-1})\in\miss_\psi(w_{l})$. Assume we have two clients $v, v'$ such that $\psi(u, w_k)\in \miss_\psi(v), \psi(u, w_{k'})\in \miss_\psi(v')$. The key observation is that we can virtually rotate colors around $u$ to move the uncolored edges $(u, v), (u, v')$ to $(u, w_k), (u, w_{k'})$. Then, as $w_{k'}, w_k$ are both missing the same color $\psi(u, w_{l})$, we are actually back to the aforementioned \emph{pairing the clients case}, where we can pair $w_k, w_{k'}$ together by considering the $\{x, \psi(u, w_{l})\}$-alternating paths from $w_k, w_{k'}$, respectively. See \Cref{overview-star} for an illustration.

In our actual algorithm, we will not actually rotate both uncolored edges $(u, v), (u, v')$, but merely pick a random client $v$ to extend the coloring $\psi$ while pretending all other uncolored edges $(u, v')$ were rotated. In other words, such rotations of all uncolored edges like $(u, v), (u, v')$ are only {\em virtually done in the analysis}, while the algorithm only rotates a random one to extend one more color. Step \ref{alg-step3} of our main algorithm will be dealing with this case.

\begin{figure}
	\centering
	\input{figs/overview-star}
	\caption{If the neighborhood of $u$ is a star rather than a chain, then we could virtually rotate uncolored edges to $(u, w_0), (u, w_1)$. As $w_0, w_1$ are both missing blue around them, we can
pair them together and reuse the argument for the
 \emph{pairing the clients case}.}
 %pair them together and reuse the argument for the bipartite case.}
 \label{overview-star}
\end{figure}
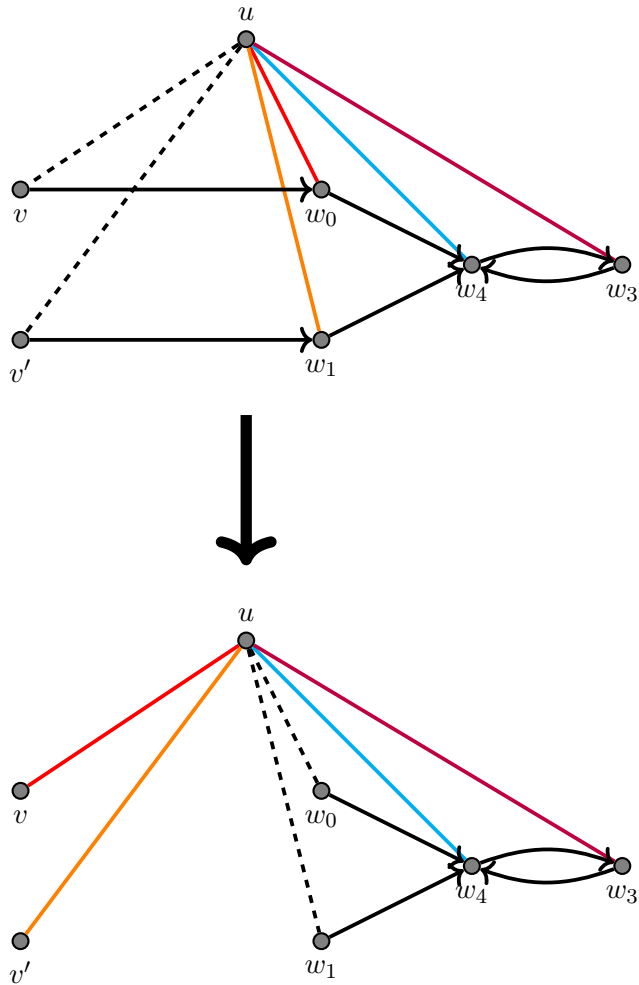

\paragraph{Case IV: Congestion control.} \label{overview-step4} One final obstacle comes from the interference between different star center $u$'s. When applying the above technique for handling star-like neighborhoods of both $u$ and $u'$, we might wish to virtually move uncolored edges $(u, v)$ and $(u', v')$ to a common neighbor $w$ of $u, u'$. In this case, when averaging the length of alternating path, the one starting at $w$ may be counted too many times --- an issue which will be called \emph{congestion} (see \Cref{def-cong}). To alleviate this disadvantage, letting $y$ be the missing color around $w$, our key observation is that $u$ and $u'$ have neighbors $w_l$ and $w_{l'}$, respectively, such that $y=\psi(u, w_l) = \psi(u', w_{l'})$. Then, we can consider $\{x, y\}$-alternating paths starting from $u, u'$ respectively; in other words, even when $w$ is counted many times, we can still collect a large number of different $\{x, y\}$-alternating paths of the same type. See \Cref{overview-cong} for an illustration. Step \ref{alg-step4} of our main algorithm will be dealing with this case.

\begin{figure}
	\centering
	\input{figs/overview-cong}
	\caption{If $w$ is the target of many clients in different star subgraphs, we can show that all these star centers have the same type of alternating paths which begin with blue edges.}\label{overview-cong}
\end{figure}
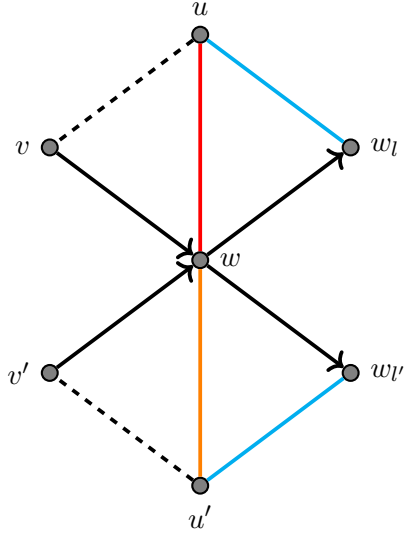

%% file: figs/overview-fan.tex
\begin{tikzpicture}[thick,scale=1]
	\draw (2, 2) node(1)[circle, draw, fill=black!50,
	inner sep=0pt, minimum width=6pt, label = $u$] {};
	
	\draw (-3, 0) node(3)[circle, draw, fill=black!50,
	inner sep=0pt, minimum width=6pt, label = {-90: {$v$}}] {};
	
	\draw (-1, 0) node(4)[circle, draw, fill=black!50,
	inner sep=0pt, minimum width=6pt, label = {-90: {$v'$}}] {};

	\draw (1, 0) node(5)[circle, draw, fill=black!50,
	inner sep=0pt, minimum width=6pt, label = {-90: {$w_0$}}] {};
	
	\draw (3, 0) node(6)[circle, draw, fill=black!50,
	inner sep=0pt, minimum width=6pt, label = {-90: {$w_1$}}] {};
	
	\draw (5, 0) node(7)[circle, draw, fill=black!50,
	inner sep=0pt, minimum width=6pt, label = {-90: {$w_2$}}] {};
	
	\draw (7, 0) node(8)[circle, draw, fill=black!50,
	inner sep=0pt, minimum width=6pt, label = {-90: {$w_3$}}] {};
	
	\draw [line width = 0.5mm, dashed] (1) to (3);
	\draw [line width = 0.5mm, dashed] (1) to (4);
	\draw [line width = 0.5mm, color=red] (1) to (5);
	\draw [line width = 0.5mm, color=orange] (1) to (6);
	\draw [line width = 0.5mm, color=cyan] (1) to (7);
	\draw [line width = 0.5mm, color=purple] (1) to (8);
	\draw [->, line width = 0.5mm] (5) to (6);
	\draw [->, line width = 0.5mm] (6) to (7);
	\draw [->, line width = 0.5mm] (7) to[out=20, in=160] (8);
	\draw [->, line width = 0.5mm] (8) to[out=200, in=-20] (7);
	\draw [->, line width = 0.5mm] (3) to[out=-40, in=-140] (6);
	\draw [->, line width = 0.5mm] (4) to[out=-40, in=-140] (5);
	
\end{tikzpicture}

%% file: figs/overview-star.tex
\begin{tikzpicture}[thick,scale=1]
	\draw (0, 2) node(1)[circle, draw, fill=black!50,
	inner sep=0pt, minimum width=6pt, label = $u$] {};
	
	\draw (-3, 0) node(2)[circle, draw, fill=black!50,
	inner sep=0pt, minimum width=6pt, label = {-90: {$v$}}] {};
	
	\draw (-3, -2) node(3)[circle, draw, fill=black!50,
	inner sep=0pt, minimum width=6pt, label = {-90: {$v'$}}] {};

	\draw (1, 0) node(4)[circle, draw, fill=black!50,
	inner sep=0pt, minimum width=6pt, label = {-90: {$w_0$}}] {};
	
	\draw (1, -2) node(5)[circle, draw, fill=black!50,
	inner sep=0pt, minimum width=6pt, label = {-90: {$w_1$}}] {};
	
	\draw (3, -1) node(6)[circle, draw, fill=black!50,
	inner sep=0pt, minimum width=6pt, label = {-90: {$w_4$}}] {};
	
	\draw (5, -1) node(7)[circle, draw, fill=black!50,
	inner sep=0pt, minimum width=6pt, label = {-90: {$w_3$}}] {};
	
	\draw [line width = 0.5mm, dashed] (1) to (2);
	\draw [line width = 0.5mm, dashed] (1) to (3);
	
	\draw [line width = 0.5mm, color=red] (1) to (4);
	\draw [line width = 0.5mm, color=orange] (1) to (5);
	\draw [line width = 0.5mm, color=cyan] (1) to (6);
	\draw [line width = 0.5mm, color=purple] (1) to (7);
	
	\draw [->, line width = 0.5mm] (2) to (4);
	\draw [->, line width = 0.5mm] (3) to (5);
	\draw [->, line width = 0.5mm] (4) to (6);
	\draw [->, line width = 0.5mm] (5) to (6);
	\draw [->, line width = 0.5mm] (6) to[out=20, in=160] (7);
	\draw [->, line width = 0.5mm] (7) to[out=200, in=-20] (6);
	
	\draw [->, line width = 1.5mm] (0, -3) to (0, -5);
	%==================================
	
	\draw (0, -6) node(8)[circle, draw, fill=black!50,
	inner sep=0pt, minimum width=6pt, label = $u$] {};
	
	\draw (-3, -8) node(9)[circle, draw, fill=black!50,
	inner sep=0pt, minimum width=6pt, label = {-90: {$v$}}] {};
	
	\draw (-3, -10) node(10)[circle, draw, fill=black!50,
	inner sep=0pt, minimum width=6pt, label = {-90: {$v'$}}] {};

	\draw (1, -8) node(11)[circle, draw, fill=black!50,
	inner sep=0pt, minimum width=6pt, label = {-90: {$w_0$}}] {};
	
	\draw (1, -10) node(12)[circle, draw, fill=black!50,
	inner sep=0pt, minimum width=6pt, label = {-90: {$w_1$}}] {};
	
	\draw (3, -9) node(13)[circle, draw, fill=black!50,
	inner sep=0pt, minimum width=6pt, label = {-90: {$w_4$}}] {};
	
	\draw (5, -9) node(14)[circle, draw, fill=black!50,
	inner sep=0pt, minimum width=6pt, label = {-90: {$w_3$}}] {};
	
	\draw [line width = 0.5mm, color=red] (8) to (9);
	\draw [line width = 0.5mm, color=orange] (8) to (10);
	
	\draw [line width = 0.5mm, dashed] (8) to (11);
	\draw [line width = 0.5mm, dashed] (8) to (12);
	\draw [line width = 0.5mm, color=cyan] (8) to (13);
	\draw [line width = 0.5mm, color=purple] (8) to (14);
	
	%\draw [->, line width = 0.5mm] (9) to (11);
	%\draw [->, line width = 0.5mm] (10) to (12);
	\draw [->, line width = 0.5mm] (11) to (13);
	\draw [->, line width = 0.5mm] (12) to (13);
	\draw [->, line width = 0.5mm] (13) to[out=20, in=160] (14);
	\draw [->, line width = 0.5mm] (14) to[out=200, in=-20] (13);
\end{tikzpicture}

%% file: figs/overview-cong.tex
\begin{tikzpicture}[thick,scale=1]
	\draw (0, 3) node(1)[circle, draw, fill=black!50,
	inner sep=0pt, minimum width=6pt, label = $u$] {};
	\draw (-2, 1.5) node(2)[circle, draw, fill=black!50,
	inner sep=0pt, minimum width=6pt, label = {180: {$v$}}] {};
	\draw (2, 1.5) node(3)[circle, draw, fill=black!50,
	inner sep=0pt, minimum width=6pt, label = {0:{$w_l$}}] {};
	\draw (0, 0) node(4)[circle, draw, fill=black!50,
	inner sep=0pt, minimum width=6pt, label = {0:{$w$}}] {};
	\draw (0, -3) node(5)[circle, draw, fill=black!50,
	inner sep=0pt, minimum width=6pt, label = {-90:{$u'$}}] {};
	\draw (-2, -1.5) node(6)[circle, draw, fill=black!50,
	inner sep=0pt, minimum width=6pt, label = {180:{$v'$}}] {};
	\draw (2, -1.5) node(7)[circle, draw, fill=black!50,
	inner sep=0pt, minimum width=6pt, label = {0:{$w_{l'}$}}] {};
	
	\draw [line width = 0.5mm, dashed] (1) to (2);
	\draw [line width = 0.5mm, color=cyan] (1) to (3);
	\draw [line width = 0.5mm, color=red] (1) to (4);
	\draw [line width = 0.5mm, dashed] (5) to (6);
	\draw [line width = 0.5mm, color=cyan] (5) to (7);
	\draw [line width = 0.5mm, color=orange] (5) to (4);
	
	\draw [->, line width = 0.5mm] (2) to (4);
	\draw [->, line width = 0.5mm] (4) to (3);
	\draw [->, line width = 0.5mm] (6) to (4);
	\draw [->, line width = 0.5mm] (4) to (7);
	
\end{tikzpicture}

%% file: prelim.tex
\section{Preliminaries}\label{prelim}
All logarithms will take base $2$. Given a simple undirected graph $G = (V, E)$ on $n$ vertices and $m$ edges with maximum degree $\Delta\geq n^{1/3}$, for any vertex subset $S$, let $E[S]$ be all edges in $E$ between vertices in $S$, and for any pair of disjoint vertex sets $X, Y$, let $E[X, Y]$ be the set of edges in $E$ between $X, Y$.

If $\Delta < n^{1/3}$, then we can apply the $\tilde{O}(m\Delta)$ time algorithm from \cite{gabow1985algorithms}. Let $\psi: E\rightarrow \{\bot, 1, 2, \ldots, \Delta+1\}$ be a valid {\em partial} edge coloring. For any vertex $v\in V$, let $N_\psi(v) = N_G^\psi(v)$ be the set of  { uncolored neighbors} adjacent to $v$,
and write $\deg_G^\psi(v) = |N_G^\psi(v)|$;
a neighbor $v$ of $u$ is called an {\em uncolored neighbor} of $u$, if $\psi(u, v) = \bot$. 
Also, let $\miss_\psi(v)\subseteq \{1, 2, \ldots, \Delta+1\}$ be the set of colors that are {\em missing} around vertex $v$.

For any pair of distinct colors $x, y\in \{1, 2, \ldots, \Delta+1\}$, an {\em $\{x, y\}$-alternating path} is a maximal simple path $P = \langle u_0, u_1, \ldots, u_k\rangle$ in $G$ whose edges are receiving colors from $\{x, y\}$ such that $u_0, u_k$ are missing exactly one color in $\{x, y\}$; the number of edges $k$ in any path $P$ is denoted by $|P|$. A {\em flip} operation of an $\{x, y\}$-alternating path $P$ is to exchange the $x, y$ colors of all edges along $P$.

\begin{lemma}\label{sum-alt-path}
	Given any partial edge coloring $\psi$ and any color $x\in \{1, 2, \ldots, \Delta+1\}$, for any $y\in \{1, 2, \ldots, \Delta+1\}\setminus \{x\}$, let $L_{x, y}$ denote the total length of all $\{x, y\}$-alternating paths of lengths at least $2$. Then $\sum_{y\neq x}L_{x, y} < 3m$.
\end{lemma}
\begin{proof}
	For any $\{x, y\}$-alternating path $P$ such that $|P|\geq 2$, the number of edges in $P$ with color $x$ is at most twice the number of edges in $P$ with color $y$. Therefore, taking a summing over all $y\neq x$ and all such $\{x, y\}$-alternating paths, we have 
	$$\sum_{y\neq x}L_{x, y} \leq 3\cdot |\{e\in E\mid \psi(e)\neq x\}| < 3m.$$
\end{proof}

\begin{lemma}[Eulerian partition 
\cite{arjomandi1982efficient,gabow1985algorithms,sinnamon2019fast}]\label{euler}
	Given an undirected simple graph $G = (V, E)$ on $n$ vertices and $m$ edges and maximum degree $\Delta$, there is a linear time deterministic algorithm that partitions $G$ into to two edge-disjoint subgraphs $G_1 = (V, E_1), G_2 = (V, E_2)$, such that the maximum degrees $\Delta_1, \Delta_2$ of $G_1, G_2$ are in the range $\floor{\Delta/2}\leq \Delta_1, \Delta_2\leq \ceil{\Delta/2}$.
\end{lemma}

\subsection{Vizing's Original Algorithm} Let us describe the basic procedure by \cite{vizing1965critical} that extends any partial edge coloring $\psi$ by one more colored edge. Let $(u, v)\in E$ be any uncolored edge under $\psi$; that is, $\psi(u, v) = \bot$. Then find a sequence of distinct neighbors $v = v_0, v_1,v_2, \ldots, v_k$ of $u$ such that the following holds; this sequence $v_1,v_2, \ldots, v_k$ is usually called a {\em Vizing fan}.
\begin{itemize}[leftmargin=*]
	\item For any $1\leq i\leq k$, $\psi(u, v_i)\in \miss_\psi(v_{i-1})$.
	\item Either $\miss_\psi(u)\cap \miss_\psi(v_k)\neq \emptyset$, or there exists an index $1\leq j<k$ such that $\psi(u, v_j)\in \miss_\psi(v_k)$.
\end{itemize}

If $\miss_\psi(u)\cap \miss_\psi(v_k)\neq \emptyset$, then take an arbitrary color $x\in \miss_\psi(u)\cap \miss_\psi(v_k)$ and rotate the coloring around $u$ as: $\psi(u, v_i)\leftarrow \psi(u, v_{i+1}), 0\leq i<k$, and $\psi(u, v_k)\leftarrow x$.

Now, let us assume $\miss_\psi(u)\cap \miss_\psi(v_k) =  \emptyset$, so there must exist an index $1\leq j<k$ such that $\psi(u, v_j)\in \miss_\psi(v_k)$. Take an arbitrary color $x\in \miss_\psi(u)$, and define $y = \psi(u, v_j)$. Let $P$ be the $\{x, y\}$-alternating path beginning at $u$.
\begin{enumerate}[(1),leftmargin=*]
    \item $P$ does not end at $v_{j-1}$.

    Apply a rotation operation: $\psi(u, v_i)\leftarrow \psi(u, v_{i+1}), 0\leq i<j$, and flip the $\{x, y\}$-alternating path $P$. Finally, assign $\psi(u, v_{j})\leftarrow x$.

    \item $P$ ends at $v_{j-1}$.
	
    Flip the color of the $\{x, y\}$-alternating path from $u$, then apply a rotation operation: $\psi(u, v_i)\leftarrow \psi(u, v_{i+1}), 0\leq i<k$, then assign color $\psi(u, v_k)\leftarrow y$.
\end{enumerate}

The runtime of Vizing's procedure is bounded by $\tilde{O}(\Delta + |P|)$.

\subsection{Fast Edge Coloring with a Large Additive Slack}
We prove the following statement, which will be used in our main algorithm; a similar statement was proved in \cite{elkin2024deterministic}. Without this lemma, we can still obtain a polynomial improvement over the time barrier $\tilde{O}(m\sqrt{n})$ of $(\Delta+1)$-edge coloring, by instead using the near-linear time $(\Delta+\tilde{O}(\sqrt{\Delta}))$-edge coloring algorithm from \cite{karloff1987efficient}.

\begin{lemma}[\cite{elkin2024deterministic}]\label{slack}
	Given an undirected simple graph $G = (V, E)$ on $n$ vertices and $m$ edges with maximum degree $\Delta$. For any integer $d<\Delta$, there is a deterministic algorithm that computes a $(\Delta + d)$-edge coloring of $G$ in $\tilde{O}(m\Delta / d)$ time.
\end{lemma}

\begin{proof}  
	Assume that $d \in [2^l, 2^{l+1})$. The algorithm is by a recursive application of the Eulerian partition provided by \Cref{euler}. Basically, on input tuple $\brac{G_0 = (V, E_0), \Delta_0}$, if $\Delta_0 < \Delta / 2^{l-3}+1$, then simply run the algorithm from \cite{gabow1985algorithms} to find a $(\Delta_0+1)$-edge coloring in time $\tilde{O}(|E_0|\Delta / d)$. Otherwise, apply \Cref{euler} on $G_0$, which produces a partition $G_0 = G_1\cup G_2$ in linear time $O(|E_0|)$ into subgraphs $G_1$ and $G_2$ with maximum degrees $\floor{\Delta_0/2}\leq \Delta_1, \Delta_2\leq \ceil{\Delta_0/2}$, respectively. After that, recursively compute two edge colorings $\psi_1, \psi_2$ on tuples $\brac{G_1, \Delta_1}$ and $\brac{G_2, \Delta_2}$, respectively, with disjoint color palettes, and merge the two colorings into a single coloring $\psi_0 = \psi_1\cup \psi_2$ for $G_0$.
	
	The runtime due to each level of recursion, excluding the leaves of the recursion tree, is the time it takes to run the Eulerian partition over all subgraphs at that level, namely $O(m)$; since there are at most $O(\log \Delta)$ recursion levels, this gives a runtime of $\tilde{O}(m)$.
The applications of the algorithm from \cite{gabow1985algorithms} at the leaves of the recursion trees take a total time of $\tilde{O}(m\Delta / d)$. As for the number of colors, it is shown inductively that if graph instance $G_0$ appears on a recursion tree node with depth $i$, then its maximum degree $\Delta_0$ satisfies $\Delta_0 < \frac{\Delta}{2^i} + 1$. Therefore, the recursion tree has depth at most $l-3$, thus at most $2^{l-3}$ different leaves. Since each leaf node uses at most $\frac{\Delta}{2^{l-3}}+2$ colors, the total number of different colors used for the edges of $G$ is at most $\Delta + 2^{l-2} < \Delta+d$.
\end{proof}

%% file: alg.tex
\section{Faster Edge-Coloring (Proof of Theorem~\ref{faster-vizing})}
We can assume the maximum degree $\Delta$ is at least $\Omega(n^{1/3})$, since otherwise we can use the algorithm from \cite{gabow1985algorithms} to compute a $(\Delta+1)$-edge coloring in $\tilde{O}(m\Delta) = \tilde{O}(mn^{1/3})$ time. If $\Delta > n^{2/3}$, then we will apply \Cref{euler} to reduce to the case of $\Delta = n^{2/3}$ (see \Cref{large-Delta}).

The main part of this section is devoted to the proof of the following statement. 

\begin{theorem}\label{faster-vizing-small}
	Given a simple undirected graph $G = (V, E)$ on $n$ vertices and $m$ edges with maximum degree $\Delta\in \left[ \tilde{\Omega}(n^{1/3}), n^{2/3}\right]$, there is a randomized algorithm that computes a $(\Delta+1)$-edge coloring with high probability in runtime $\tilde{O}\brac{mn^{\frac{1}{3}}}$.
\end{theorem}

Let $V_\lo = \{v\in V\mid \deg_G(v)\leq \Delta /3\}$, and $V_\hi = V\setminus V_\lo$, and denote $n_0 = |V_\hi|$. Take two random subsets of vertices $U\subseteq V_\hi, W\subseteq V_\lo$ which includes each vertex independently with probability $\frac{100\log n}{n^{1/3}}$; if $n_0 < n^{1/3} < \Delta/3$, then simply set $U = \emptyset$.

{\bf Remark.~} If the readers only want to understand how the algorithm behaves on near-regular graphs, they can simply assume $V_{\lo} = \emptyset$ and set $m = n\Delta$ whenever we encounter a complicated expression.

\begin{claim}\label{hitset}
	With high probability the following three statements hold:
	\begin{itemize}[leftmargin=*]
		\item the maximum degree of graph $$G_0 = \brac{V, E[U]\cup E[V_\hi\setminus U]\cup E[V_\hi, V_\lo\setminus W]\cup E[V_\lo]}$$ is at most $\Delta - 10\Delta / n^{1/3}$.
		\item $|U|\leq \frac{200n_0\log n}{n^{1/3}}\leq \frac{200m\log n}{\Delta n^{1/3}}$, and $\sum_{v\in W}\deg_G(v)\leq \frac{200m\log n}{n^{1/3}}$.
	\end{itemize}
\end{claim}
\begin{proof}
	The second property is straightforward by standard concentration inequalities. So let us focus on the first property.
	
	Consider any vertex $v\in V_\hi$. If $v$ has less than $\Delta/6$ neighbors in $V_\hi$, then it has more than $\Delta/6$ neighbors in $V_\lo$. Then, with high probability, it has at least $10\Delta/n^{1/3}$ neighbors in $W$. As $G_0$ excludes all edges in $E[W, V_\hi]$, the degree of $v$ in $G_0$ is at most $\Delta - 10\Delta / n^{1/3}$.
	
	Otherwise, it has $\Delta/6$ neighbors in $V_\hi$. If $v\in V_\hi\setminus U$, then with high probability it has $10\Delta/n^{1/3}$ neighbors in $U$. As $G_0$ excludes all edges in $E[U, V_\hi\setminus U]$, the degree of $v$ in $G_0$ is at most $\Delta - 10\Delta / n^{1/3}$.
	
	If $v\in U$, then with high probability it has at most $\frac{500\Delta\log^2 n}{n^{1/3}} < \Delta/6 - 10\Delta / n^{1/3}$ neighbors in $U$, and thus at least $10\Delta / n^{1/3}$ neighbors in $V_\hi \setminus U$. As $G_0$ excludes all edges in $E[U, V_\hi\setminus U]$, the degree of $v$ in $G_0$ is at most $\Delta - 10\Delta / n^{1/3}$.
\end{proof}

%\subsection{Coloring high-degree edges}
Given any partial coloring of $G$, for any $u\in V$, take an arbitrary color $c_\psi(u)\in \miss_\psi(u)$. Let $C_\psi(u)$ be the set of neighbors such that $\psi(u, v)\neq \bot, \forall v\in C_\psi(u)$. For any set $C_\psi(u)$, build a directed graph $T_\psi(u)$ consisting of all vertices in $C_\psi(u)$, and for each pair of vertices $v, w\in C_\psi(u)$, add a directed edge $(v, w)$ if $c_\psi(v) = \psi(u, w)$. Then each vertex in $T_\psi(u)$ has out-degree at most one; if some vertices have out-degree $0$, we can easily extend the coloring later on. Hence, each weakly connected component in $T_\psi(u)$ is a directed tree plus at most one extra edge. If we treat this weakly connected component as an undirected graph, and remove an arbitrary edge on the unique fundamental cycle, then we would get a tree; subsequently, we will be using the same notation $T_{\psi}(*)$ to denote the tree structures {\em after} removing one potential edge. Formally, we have the following statement.

\begin{claim}
    For each weakly connected component in $T_\psi(u)$, we can find a spanning tree component where each vertex has a directed edge towards its parent.
\end{claim}

Let us proceed to color all edges in $E$ incident on $V_\hi$; that is, all edges in $E \setminus E[V_\lo]$. Under the event that \Cref{hitset} holds, we know that the maximum degree of $G_0$ is at most $\Delta - 10\Delta / n^{1/3}$. Then, applying \Cref{slack} on $G_0$, we can find a $(\Delta+1)$-edge coloring $\psi$ for all edges in $G_0$ in $\tilde{O}(m n^{1/3})$ time. Next, the main task is to extend $\psi$ to all edges in $E[U, V_\hi\setminus U]\cup E[W, V_\hi]$. Our algorithm repeats the following procedure for multiple phases until all edges in $E[U, V_\hi\setminus U]\cup E[W, V_\hi]$ are colored under $\psi$.

\paragraph{Preparation.} We will color $E[U, V_\hi\setminus U]$ and $E[W, V_\hi]$ using the same algorithm, but with different settings of parameters. The coloring process will consist of multiple phases. In each phase, let $m_0$ be the current number of uncolored edges in $E[U, V_\hi\setminus U]$, and let $m_1$ be the current number of uncolored edges in $E[W, V_\hi]$. 

If $m_0 = \tilde{O}\brac{m / n^{2/3}}$, then we will apply Vizing's procedure for each of the uncolored edges, which will take time $\tilde{O}\brac{mn^{1/3}}$. Similarly, if $m_1 = \tilde{O}(m/n)$, then we will apply Vizing's procedure for each of the uncolored edge, which will take near-linear time.

For the rest, we are going to assume $m_0 = \tilde{\Omega}\brac{m/n^{2/3}}$ or $m_1 = \tilde{\Omega}(m/n)$. By the above preparatory steps, we can claim the following lower bounds on $m_0, m_1$.
\begin{claim}\label{uncolored-lb}
	If $m_0 \neq 0$, then $m_0 =\tilde{\Omega}\brac{m/n^{2/3}}$; if $m_1\neq 0$, then $m_1 = \tilde{\Omega}(m/n)$.
\end{claim}

The algorithm will perform several rounds to extend $\psi$ until the number of uncolored edges under $\psi$ in $E[U, V_\hi\setminus U]\cup E[W, V_\hi]$ has dropped below $\frac{3}{4}(m_0 + m_1)$. In each round, define:
$$R = \begin{cases}
	U	&	m_0\geq m_1\\
	W	&	m_0 < m_1
\end{cases}$$

Let $d$ be the integer power of $2$ such that the total degree: $$\sum_{\deg_G^\psi(u)\in [d, 2d), u\in R}\deg^\psi_G(u)$$ is maximized. Then, define $R_d = \{u\mid \deg^\psi_G(u)\in [d, 2d), u\in R\}$, and pick the color $x\in \{1, 2, \ldots, \Delta+1\}$ such that $\sum_{u\in X}\deg_G^\psi(u)$ is maximized, where $X\subseteq R_d$ is the set of all vertices $v$ such that $\miss_\psi(v)\ni x$. 

\begin{claim}\label{uncolored}
    The value of $\sum_{u\in X}\deg_G^\psi(u)$ is at least: $$\sum_{u\in X}\deg_G^\psi(u)\geq\begin{cases}\frac{m_0^2}{2\Delta|U|\log^2 n}\brac{\geq \frac{m_0^2n^{1/3}}{400m\log^3 n}}	&	m_0\geq m_1\\
		\frac{m_1}{2\log n}	&	m_0 < m_1
	\end{cases}$$
	Plus, $|X|\geq \frac{\max\left\{m_0, m_1\right\} }{4\Delta\log n} $.
\end{claim}

\begin{proof}
	If $m_0\geq m_1$, then by the maximizing condition, we have:
	$$2d|U|\geq \sum_{u\in R_d}\deg_G^\psi(u)\geq \frac{m_0}{\log n}$$
	Since $|\miss_\psi(u)|\geq d+1, \forall u\in R_d$, by an averaging argument, we have:
	
	$$\begin{aligned}
		2d|X|\geq \sum_{u\in X}\deg_G^\psi(u)&\geq \frac{1}{\Delta+1}\sum_{z=1}^{\Delta+1}\sum_{u\in R_d, \miss_\psi(u)\ni z}\deg_G^\psi(u)\\
		&\geq \frac{d+1}{\Delta+1}\sum_{u\in R_d}\deg_G^\psi(u)
		\geq \frac{m_0d}{\Delta\log n} > \frac{m_0^2}{2\Delta|U|\log^3 n}
	\end{aligned}$$
	This also yields $|X|\geq m_0 / 2\Delta\log n$ at the same time.

	If $m_0 < m_1$, then by the maximizing condition, we have:
	$$\sum_{u\in R_d}\deg_G^\psi(u)\geq \frac{m_1}{\log n}$$
	
	Since for any $u\in R_d\subseteq V_\lo$ we have $|\miss_\psi(u)|> 2\Delta/3$, an averaging argument yields:
	$$\begin{aligned}
		2d|X|\geq \sum_{u\in X}\deg_G^\psi(u)&\geq\frac{1}{\Delta+1}\sum_{z=1}^{\Delta+1}\sum_{u\in R_d, \miss_\psi(u)\ni z}\deg_G^\psi(u)\\
		&\geq \frac{2\Delta/3}{\Delta+1}\sum_{u\in R_d}\deg_G^\psi(u) > \frac{m_1}{2\log n}\end{aligned}$$
	In this case, it also gives $|X|\geq m_1 / 4\Delta\log n$ at the same time.
\end{proof}

\paragraph{Color extension.} We are going to iteratively perform a sequence of color extensions on edges of $X$ in $G$, and each time an uncolored edge incident on a vertex $u\in X$ receives color $x$, we will remove $u$ from $X$. This iterative process terminates when $|X|$ has dropped by half. Throughout this iterative color extensions, we will build and maintain the following data structures; the readers can ignore the $T_{\psi}(.)$ notation if they are happy to pretend that we do not need to resolve the Vizing fans.

\begin{framed}
	Data structure $\ds_X$ for one round of color extensions.
\begin{itemize}[leftmargin=*]
	\item As mentioned in the preliminaries, each vertex $v\in V$ maintains a special color $c_\psi(v)\in \miss_\psi(v)$ which is used to compute the directed graphs $T_\psi(*)$.
	
	This data structure is maintained throughout the whole algorithm, so we do not need to initialize this part at the beginning of each round.
	
	\item For vertex $v\in V_\hi$, if $v$ has an uncolored edge connecting to $X$, let $u_1, u_2, \ldots, u_l$ be all uncolored neighbors of $v$ in $X$. For each $u_i$, associate the edge $(u_i, v)$ with a color $\clr_\psi(v\rightarrow u_i)\in \miss_\psi(v)$, such that all colors $\clr_\psi(v\rightarrow u_i)$ are different for $1\leq i\leq l$.
	
	\item For each vertex $u\in X$ and each color $y\in \{1, 2, \ldots, \Delta+1\}$, maintain a list of neighbors: $$\lst_\psi(u, y) = \{v\mid y = \clr_\psi(v\rightarrow u) \}.$$
\end{itemize}
\end{framed}

\begin{definition}
    For any $u\in X$ and for any directed tree component $T$ of $T_\psi(u)$ as well a vertex $v$ in $T$, $v$ is called \emph{active} if there exists a (non-strict) descendant $w$ of $v$ satisfying $|\lst_\psi\brac{u, \psi(u, w)}| = 1$.
\end{definition}

Intuitively, $v$ is active in $T_\psi(u)$ if we can make a color rotation around the neighbors of $u$ which transfers an uncolored edge to edge $(u, v)$.

\begin{definition}\label{def-branch}	
    For any $u\in X$ and for any directed tree component $T$ of $T_\psi(u)$, the \emph{branch number} $\br_T(b)$ of vertex $b\in T$ is the number of active children $v$ of $b$ in $T$. A vertex $b$ is called branching if $\br_T(b)\geq 2$.
\end{definition}

This branch number will be mainly used in the analysis of Steps \ref{alg-step3}\ref{alg-step4}.

Define a length parameter 
\begin{equation}\label{len-param}
	L = \begin{cases}
	\frac{mn^{1/3}\log^{10}n}{m_0}	&	m_0\geq m_1\\
	\frac{m^{1/2}n^{1/2}\log^{10}n}{m_1^{1/2}}	&	m_0 < m_1
\end{cases}
\end{equation}
which will be used to bound the length of alternating paths, as we wish to color uncolored edge in $\tilde{O}(\Delta+L)$ time. Next, our algorithm will repeatedly pick a uniformly random vertex $u\in X$, and build the graph $T_\psi(u)$ which takes $O(\Delta)$ time since it is a sparse digraph. Then, pick a random uncolored neighbor $v$ of $u$ and let $y = \clr_\psi(v\rightarrow u)$. Next, try several different steps to extend the color of $\psi$ around $u$; we say that a step fails if it does not extend $\psi$. The readers can only focus on Step (1) if they are happy to ignore Vizing fans.

\begin{enumerate}[(1),leftmargin=*]
	\item\label{alg-step1} If $y\in \miss_\psi(u)$, then assign $\psi(u, v)\leftarrow y$.
	
	Otherwise, let $P_{x, y}$ be the $\{x, y\}$-alternating path starting from $v$. If $|P_{x, y}|\leq L$ and $P_{x, y}$ does not end at $u$, then flip the path $P_{x, y}$, and extend the coloring $\psi(u, v)\leftarrow x$. Note that $P_{x, y}$ can be computed by tracing the alternating path from the starting point, and once we detect that $|P_{x, y}| > L$ we can terminate the search.
	
	\item\label{alg-step2} If Step (1) failed and $|\lst_\psi(u, y)|=1$, then find the unique vertex $w\in C_\psi(u)$ such that $\psi(u, w) = y$. Locate the position of $w$ in the tree structure $T_\psi(u)$. There are two cases below.
    \begin{enumerate}[(a),leftmargin=*]
	\item If $w$ is the only vertex in the tree component $T$ such that $|\lst_\psi(u, \psi(u, w))| = 1$, then just apply the standard Vizing's procedure (as described in \Cref{prelim}) to color the edge $(u, v)$. Vizing's procedure will either find an $\{x, z\}$-alternating path $Q_{x, z}$ starting from $u$, or extend $\psi$ to $(u, v)$ by rotating colors around $u$. Note that the execution of Vizing's procedure should be in accordance with special colors $c_\psi(*)$, so that the Vizing fan only involve vertices in the tree component $T$.

	If $|Q_{x, z}| \leq L$ or $Q_{x, z}$ is not needed, then we would complete the Vizing's color extension procedure; otherwise, rollback all the color changes of Vizing's procedure.

	   \item If there is another vertex $w'\in T_\psi(u)$ which is a descendant in the tree component of $w$ in $T_\psi(u)$ such that $|\lst_\psi\brac{u, \psi\brac{u, w'}}| = 1$, then take the unique vertex $v'\in \lst_\psi\brac{u, \psi(u, w')}$, and let $w'=w_0, w_1, \ldots, w_l = w$ be the directed tree path from $w'$ to $w$ in $T_\psi(u)$. Clearly $v' \neq v$ as $\psi(u, w)\neq \psi(u, w')$.

	   Let $Q_{x, y}$ be the $\{x, y\}$-alternating path starting from vertex $u$, so it begins with edge $(u, w)$. If $|Q_{x, y}|\leq L$, then consider two cases below.
		\begin{itemize}[leftmargin=*]
			\item If $Q_{x, y}$ ends at $v$, then make the color rotation around $u$: $\psi(u, v')\leftarrow\psi(u, w_0)$, $\psi(u, w_i)\leftarrow \psi(u, w_{i+1}), 0\leq i<l$, and then flip path $Q_{x, y}$.
			\item Otherwise, flip the path $Q_{x, y}$, and assign $\psi(u, v)\leftarrow y$.
		\end{itemize}

            See \Cref{step2-1,step2-2} for an illustration.
    \end{enumerate}

    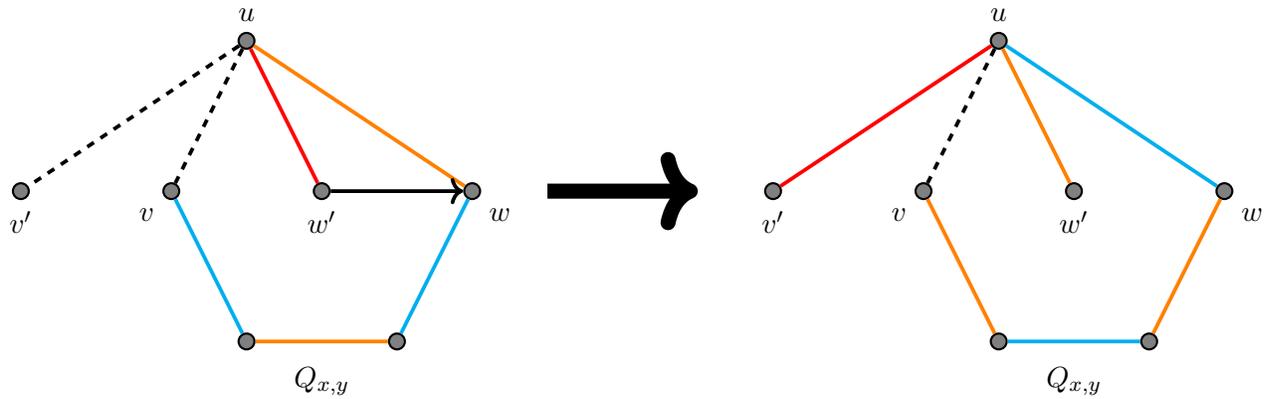
\begin{figure}
    	\centering
    	\input{figs/step2-1}
    	\caption{In this case, blue is in $\miss_\psi(u)$, orange is in $\miss_\psi(v)\cap \miss_\psi(w')$, and red is in $\miss_\psi(v')$. If $Q_{x, y}$ ends at $v$, then we can rotate the colors and extend $\psi$ to $(u, v')$. The directed edge comes from $T_\psi(u)$.}\label{step2-1}
    \end{figure}
    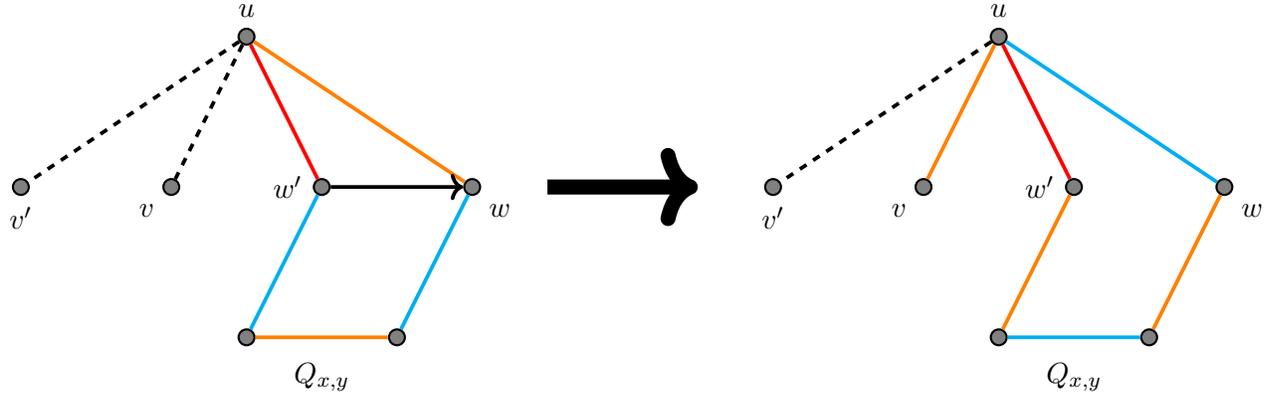
\begin{figure}
    	\centering
    	\input{figs/step2-2}
    	\caption{If $Q_{x, y}$ does not end at $v$, say $w'$, then we can extend $\psi$ to $(u, v)$.}\label{step2-2}
    \end{figure}

	\item\label{alg-step3} If Step (2) failed, then take a random index $h\in [1, \log\Delta]$, and let $B_h$ be the set of vertices in $T_\psi(u)$ such that $\br_T(b)\in [2^h, 2^{h+1})$ for any $b\in B_h$ (here $T$ is the tree component containing $b$). Then, take a vertex $b\in B_h$ uniformly at random, as well as a uniformly random active child $c$ of $b$. By definition of activeness, there exists a descendant $s$ of $c$ such that $|\lst_\psi\brac{u, \psi(u, s)}| = 1$. Pick the unique vertex $v'\in \lst_\psi\brac{u, \psi(u, s)}$, and let $s = s_0, s_1, \ldots, s_l = c$ be the directed path from $s$ to $c$.

	Let $z = \psi(u, b)$, and let $R_{x, z}$ be the $\{x, z\}$-alternating path starting at $c$. If $|R_{x, z}|\leq L$ and $R_{x, z}$ does not end at $u$, then first rotate the coloring: $\psi(u, v')\leftarrow \psi(u, s), \psi(u, s_i)\leftarrow \psi(u, s_{i+1}), 0\leq i<l$, and then flip the path $R_{x, z}$, and then assign $\psi(u, s_{l})\leftarrow x$.

    See \Cref{step3} for an illustration.
	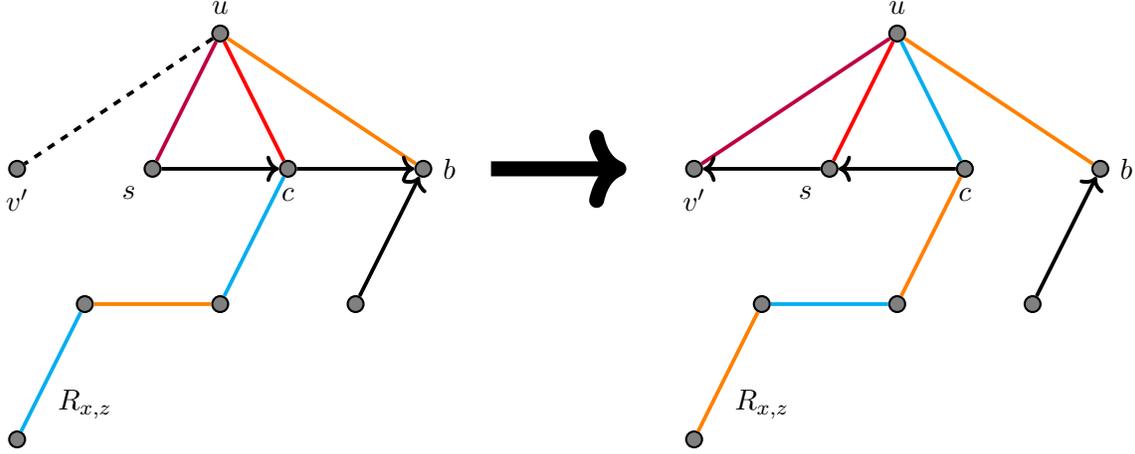
\begin{figure}
		\centering
		\input{figs/step3}
		\caption{If $R_{x, z}$ is short and does not end at $u$, then we can extend $\psi$ to $(u, v')$.}\label{step3}
	\end{figure}
 
	\item\label{alg-step4} Continue with the notations in Step (3). If Step (3) failed, then let $S_{x, z}$ be the $\{x, z\}$-alternating path starting at $u$ with edge $(u, b)$. Since $b$ is a branching vertex, it has at least two different active children $c_1, c_2$.
 
    If $|S_{s, z}|\leq L$, without loss of generality, assume that $S_{x, z}$ does not end at $c_1$. By definition of activeness, there exists a descendant $t$ of $c_1$ such that $|\lst_\psi(u, \psi(u, t))| = 1$. Pick the unique vertex $v'\in \lst_\psi(u, \psi(u, t))$, and let $t = t_0, t_1, \ldots, t_l = c_1$ be the directed path from $t$ to $c_1$. Then assign $\psi(u, v')\leftarrow \psi(u, t), \psi(u, t_i)\leftarrow \psi(u, t_{i+1}), \forall 0\leq i\leq l$, and finally flip the path $S_{x, z}$.

    See \Cref{step4} for an illustration.
	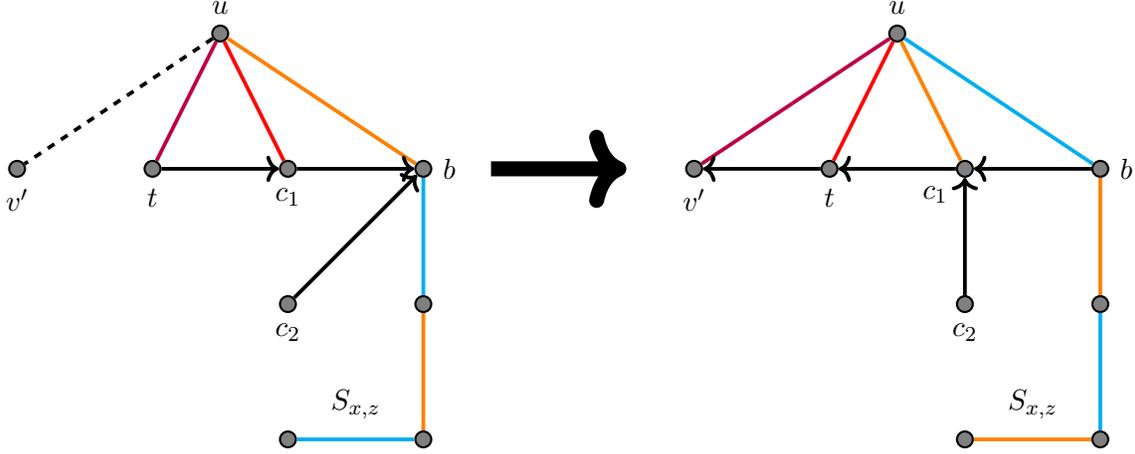
\begin{figure}
		\centering
		\input{figs/step4}
		\caption{If $S_{x, z}$ does not end at $c_1$, then we can extend $\psi$ to $(u, v')$.}\label{step4}
	\end{figure}
 
	\item In the end, if any of the color extension Step \ref{alg-step1}\ref{alg-step2}\ref{alg-step3}\ref{alg-step4} succeeded, there will be at most two vertices $u_1, u_2\in X$ around which an edge has been assigned color $x$ (note that $u$ must be one of $\{u_1, u_2\}$); in this case, remove $u_1, u_2$ from $X$ and maintain data structure $\ds_X$.
\end{enumerate}

\subsection{Runtime Analysis}

Let us first state the following lemma whose proof is deferred to the next subsection.
\begin{lemma}\label{succ}
	In each iteration, the color extension procedure succeeds with probability at least $\tilde{\Omega}(1)$.
\end{lemma}

Before we prove this lemma, let us first analyze the runtime of other parts of the algorithm.
In other words, we next provide the proof of \Cref{faster-vizing-small}, assuming the correctness of this lemma, and subsequently (see \Cref{lem41}) we prove the lemma.

\paragraph{Proof of \Cref{faster-vizing-small}.}
Let $R_d^{z}$ denote the set of vertices $\{v \in R_d \mid z \in \miss_\psi(v)$\}. At the start of each round, we want to efficiently find the color $x\in \{1, 2, \ldots, \Delta+1\}$ that maximizes the sum $\sum_{u\in R_d^x}\deg_G^\psi(u)$, as well as the set $X = R^x_d$ as in \Cref{uncolored}. To do with, we can create a data structure that, for each color $z \in \{1, \dots, \Delta+1\}$, each integer power of two $d$, and each choice $R\in \{U, W\}$, explicitly maintains the set $R_d^z$ and sum $\sum_{u\in R_d^z}\deg_G^\psi(u)$. These sets and sums, as well as $\arg \max_{z} \sum_{u\in R_d^z}\deg_G^\psi(u)$ for each integer power of $d$, can be maintained in $O(\log n)$ time each time an edge changes color and $O(\Delta \log n)$ time each time an uncolored edge receives a color. Furthermore, this data structure can be initialised in time $\tilde O(\Delta (|U| + |W|)) \leq \tilde O(mn^{1/3})$. Thus, the time spent initializing this data structure and updating it after uncolored edges recieve colors is $\tilde O(mn^{1/3}) + \tilde O(\Delta) \cdot \tilde O(m/n^{1/3}) \leq \tilde O(mn^{1/3})$, and the time spent updating the data structure each time an edge changes color leads to $O(\log n)$ overhead in the running time of our algorithm. Using this sum data structure, we can always retrieve the color $x$ in $O(\log n)$ time, and the set $X$ in $O(|X|)$ time, since it is maintained explicitly by the data structure.

% In each round, to efficiently find the color $x\in \{1, 2, \ldots, \Delta+1\}$ that maximizes the sum $\sum_{u\in R_d, \miss_\psi(u)\ni x}\deg_G^\psi(u)$ along with the set $X$ as in \Cref{uncolored}, we can explicitly maintain all the sums $\sum_{u\in R_d, \miss_\psi(u)\ni z}\deg_G^\psi(u)$ for each color $z\in \{1, 2, \ldots, \Delta+1\}$, each integer power of two $d$, and each choice $R\in \{U, W\}$. All these sums $\left\{\sum_{u\in R_d, \miss_\psi(u)\ni z}\deg_G^\psi(u)\right\}$ as well as the maximizer can be maintained in $O(\log n)$ time per edge color change to $\psi$. Using this sum data structure, we can always retrieve the maximizer $x\in \{1, 2, \ldots, \Delta+1\}$ in $O(\log n)$ time, and the set $X$ in $O(m_0 + m_1)$ time.

After $X$ has just been initialized, let $X_0$ be the screenshot of set $X$ at the moment. After that, initializing the data structure $\ds_X$ takes time $\tilde{O}(m_0+m_1)$. In each iteration, the algorithm picks a random vertex $u\in X$ and then build the directed graph $T_\psi(u)$ which takes time $\tilde{O}(\Delta)$, and the coloring steps (1)(2)(3)(4) can be carried out in $\tilde{O}(\Delta + L)$ time. If it succeeds in one more color extension, it needs to update the data structure $\ds_X$ as following.
\begin{itemize}[leftmargin=*]
	\item Firstly, at most $\Delta+2$ vertices could possibly change their color palette $\miss_\psi(*)$, including all neighbors of $u$ plus the endpoints of an alternating path. Updating the special color $c_\psi(*)$ for these vertices takes time $\tilde{O}(\Delta)$.
	
	\item Secondly, for a similar reason, maintaining all the colors $\clr_\psi(*\rightarrow *)$ takes time $\tilde{O}(\Delta)$. Plus, when the value of $\clr_\psi(v\rightarrow u)$ changes, we can update the list $\lst_\psi(u, *)$ accordingly.
\end{itemize}

By \Cref{succ}, with high probability, in time $\tilde{O}\brac{|X_0|\cdot(\Delta + L)}$, the size of $X$ would decrease by half (that is, $|X|\leq \frac{1}{2}|X_0|$), and then this round of color extension would terminate. The number of new colored edges is at least $|X_0|/4$ while the total runtime is at most $\tilde{O}\brac{m_0 + m_1 + |X_0|\cdot(\Delta + L)}$. Therefore, the amortized runtime cost of each edge color extension is bounded by 
$$\tilde{O}\brac{\frac{m_0+m_1}{|X|}+ \Delta + L}$$
By \Cref{uncolored} we already know $\frac{m_0+m_1}{|X|} = \tilde{O}(\Delta)$ which is subsumed in two other terms.

\paragraph{Case $m_0 \geq m_1$.} In this case, by \Cref{len-param} we have $L =\frac{mn^{1/3}\log^{10}n}{m_0}$. This amortization is taken within a single phase which starts with $m_0$ uncolored edges and ends with $m_0/2$ uncolored edges. In other words, the algorithm takes time $$\tilde{O}(m_0\Delta + Lm_0) = \tilde{O}\brac{\Delta^2|U| + mn^{1/3}} =\tilde{O}\brac{mn^{1/3}}$$
to half the number of uncolored edges; in the inequalities we have used the fact that $m_0\leq \Delta |U|$, $|U| = \tilde{O}(n_0 / n^{1/3})$ and $\Delta \leq n^{2/3}$ from from \Cref{hitset}.

\paragraph{Case $m_0 < m_1$.} In this case, by \Cref{len-param}  we have $L =\frac{m^{1/2}n^{1/2}\log^{10}n}{m_1^{1/2}}	$. This amortization is taken within a single phase which starts with $m_1$ uncolored edges and ends with $m_1/2$ uncolored edges. In other words, the algorithm takes time $$\tilde{O}(m_1\Delta + Lm_1) = \tilde{O}\brac{m\Delta / n^{1/3} + mn^{1/3}}$$
to half the number of uncolored edges; in the inequalities we have used the fact that $m_1 = \tilde{O}(m/n^{1/3})$ and $\Delta \leq n^{2/3}$ from \Cref{hitset}.

This completes the proof of \Cref{faster-vizing-small} under the assumption that
\Cref{succ} holds. We shall next justify this assumption.

\subsection{Proof of \Cref{succ}} \label{lem41}
By the algorithm description, for any vertex $u\in X$, once we extend one more uncolored edge around $u$ or $\miss_\psi(u)$ loses $x$ due to flips of alternating paths, $u$ would be removed from $X$. Therefore, at any moment, $\deg_G^\psi(u)\in [d, 2d), \forall u\in X$. 

Define a congestion threshold which will be used by analysis of Steps \ref{alg-step3}\ref{alg-step4}.
$$\tau = \begin{cases}
	\frac{m_0n^{2/3}}{m}	&	m_0\geq m_1\\
	\frac{m_1^{1/2}n^{1/2}}{m^{1/2}}	&	m_0 < m_1
\end{cases} $$

Note that $\tau \geq 1$ by \Cref{uncolored-lb}. Let us first show some basic inequalities regarding the parameters.
\begin{claim}\label{param-ineq}
	$\frac{m\tau}{L} < \frac{d|X|}{10^3\log n}$.
\end{claim}
\begin{proof}
	Consider two cases.
	\begin{itemize}[leftmargin=*]
		\item $m_0 \geq m_1$.
		
		In this case, by \Cref{len-param} and \Cref{uncolored} and \Cref{uncolored-lb}, we have (for $n>2^{10}$):
		$$\frac{m\tau}{L}\leq \frac{ m_0^2n^{1/3}}{m\log^{10}n}\leq \frac{m_0^2n^{1/3}}{10^4m\log^4 n} <  \frac{d|X|}{10^3\log n}$$
		
		\item $m_0 < m_1$.
		
		In this case, by \Cref{len-param} and \Cref{uncolored} and \Cref{uncolored-lb}, we have:
		$$\frac{m\tau}{L}\leq \frac{m_1}{\log^{10}n}\leq \frac{m_1}{10^4\log^2 n} <  \frac{d|X|}{10^3\log n}$$
	\end{itemize}
\end{proof}

\paragraph{Success condition of Step \ref{alg-step1}.} For a proof overview, please check \Cref{overview-step1}. Let us begin with the following definition.
\begin{definition}
	For any uncolored neighbor $v$ of $u\in X$, $v$ is called a \emph{lonely} vertex in the scope of $u$, if $\clr_\psi(v\rightarrow u)\notin \miss_\psi(u)$, and for any other uncolored neighbor $v'\neq v$ of $u$, we have $\clr_\psi(v\rightarrow u)\neq \clr_\psi(v'\rightarrow u)$. Let $\lone_\psi(u)$ counts the total number of lonely vertices in the scope of $u$.
\end{definition}

We will show that a sufficient condition for the success probability of Step (1) is:
$$\sum_{u\in X}\lone_\psi(u)\leq \frac{d|X|}{2}$$
The main intuition here is that if this condition holds, then at least a constant fraction of these alternating paths behave like in bipartite graphs.

\begin{claim}\label{step1friendly}
	For any non-lonely vertex $v$ in the scope of $u$, let $P_{x, \clr_\psi(v\rightarrow u)}(v)$ be the $\{x, \clr_\psi(v\rightarrow u)\}$-alternating path starting at $v$. If $\sum_{u\in X}\lone_\psi(u)\leq \frac{d|X|}{2}$, then there exists at least $\frac{d|X|}{8}$ many pairs of vertices $(u, v)$ such that:
	\begin{enumerate}[(i)]
		\item $v$ is not lonely in the scope of $u$;
		\item $P_{x, \clr_\psi(v\rightarrow u)}(v)$ does not end at $u$;
		\item $|P_{x, \clr_\psi(v\rightarrow u)}(v)|\leq L$.
	\end{enumerate}
	Such a pair $(u, v)$ would be called a \emph{Step-(1)-friendly} pair.
\end{claim}
\begin{proof}
	For any vertex $u$, let $v_1, v_2, \ldots, v_l, l\geq 2$ be all of $u$'s uncolored neighbors such that:
	$$y = \clr_\psi(v_1\rightarrow u) = \clr_\psi(v_2\rightarrow u) = \cdots = \clr_\psi(v_l\rightarrow u)$$ Then, at most one of the $\{x, y\}$-alternating paths $P_{x, y}(v_1), P_{x, y}(v_2), \ldots, P_{x, y}(v_l)$ could end at $u$. So, at most one alternating path $P_{x, y}(v_i)$ is terminating at $u$.
	
	If $l = 1$, then by definition of loneliness we have $y = \clr_\psi(v_1\rightarrow u)\in \miss_\psi(u)$, so the algorithm would directly assign $\psi(u, v_1)\leftarrow y$; we could view $P_{x, y}(v)$ as an empty path in this case. If $l\geq 2$, then we have $l-1\geq l/2$, and so the total number of pairs $(u, v)$ that satisfies the (i)(ii) at least $d|X|/4$.
	
	For requirement (iii), we argue that the total length of all alternating paths $P_{x, \clr_\psi(v\rightarrow u)}(v)$ such that $|P_{x, \clr_\psi(v\rightarrow u)}(v)|\geq 2$ is at most $6m$. In fact, on the one hand, since $\clr(v\rightarrow u)$ are different when $u$ changes, all alternating paths $P_{x, \clr_\psi(v\rightarrow u)}(v)$ are different for a fixed $v$. So, when taking into account of all alternating paths $P_{x, \clr_\psi(v\rightarrow u)}(v)$, each path is counted at most twice. Hence, by \Cref{sum-alt-path}, the total length of all alternating paths of $P_{x, \clr_\psi(v\rightarrow u)}(v)$ such that $|P_{x, \clr_\psi(v\rightarrow u)}(v)|\geq 2$ is bounded by $6m$.
	
	Using this upper bound and \Cref{param-ineq}, there are at most $24m / L < d|X|/8$ alternating path $P_{x, \clr_\psi(v\rightarrow u)}(v)$ satisfying the first two bullets that could possibly have length larger than $L$. In other words, among the $d|X|/4$ pairs $(u, v)$ that satisfy (i)(ii), there are at most $16\Delta n_0/L$ of them that could violate (iii). Hence, there are at least $d|X|/4 - d|X|/8\geq d|X|/8$ Step-(1)-friendly pairs.

\end{proof}

By \Cref{step1friendly}, under the condition that $\sum_{u\in X}\lone_\psi(u)\leq \frac{d|X|}{2}$, there are at least $d|X|/8$ Step-(1)-friendly pairs. If Step (1) successfully samples such a pair, then it would succeed in the color extension. Since each pair is sampled with probability at least $\frac{1}{2d|X|}$, the success probability of Step (1) is at least $1/16$.

\paragraph{Success condition of Step~\ref{alg-step2}} For a proof overview, please check \Cref{overview-step2}. Next, we assume that the success condition of Step (1) does not hold; that is $\sum_{u\in X}\lone_\psi(u) >  \frac{d|X|}{2}$.

For any pair $(u, v)$ where $u\in X$ and $v$ is an uncolored neighbor of $u$, if $v$ satisfies the first condition in \Cref{step2friend}, then let $Q_{x}(u, v)$ be the Vizing chain if we were to apply Vizing's procedure to color the edge $(u, v)$ in Step (2)(a); if Step (2)(b) were to be operated on $v$, then let $Q_{x, \clr_\psi(v\rightarrow u)}(v)$ be the $(x, \clr_\psi(v\rightarrow u))$-alternating path starting from $u$.

\begin{definition}\label{step2friend}
	For any vertex $u\in X$ as well as a lonely uncolored neighbor $v$ of $u$, let $T\subseteq T_\psi(u)$ be the tree component containing the unique vertex $w$ such that $\psi(u, w) = \clr_\psi(v\rightarrow u)$. We say that $(u, v)$ is a \emph{Step-(2)-friendly} pair if:
	\begin{enumerate}[(i)]
		\item either $w$ is the only vertex in $T$ such that $|\lst_\psi(u, \psi(u, w))| = 1$;
		\item or there exists a strict descendant $w'$ of $w$ in $T$ such that $v\notin\lst_\psi(u, \psi(u, w'))\neq \emptyset$;
		\item the length of $|Q_x(u, v)|$ or $|Q_{x, \clr_\psi(v\rightarrow u)}(v)|$ is at most $L$.
	\end{enumerate}
\end{definition}

For any $u\in X$, let $\desc_\psi(u)$ be the total number of lonely uncolored neighbors $v$ of $u$ such that $(u, v)$ satisfies property (i)(ii) in \Cref{step2friend}. 
We will show that a sufficient condition for the success probability of Step (2) is that: $$\sum_{u\in X}\desc_\psi(u)\geq \frac{d|X|}{8}$$

\iffalse
\begin{claim}\label{step2length}
	The total length of all alternating paths $$\left\{Q_{x}(u, v) \text{ or } Q_{x, \clr_\psi(v\rightarrow u)}(v)\mid (u, v)\text{ is Step-(2)-friendly}\right\}$$ is at most $4\Delta n_0$.
\end{claim}
\begin{proof}
	When $(u, v)$ is Step-(2)-friendly, by definition we know that such a vertex $v$ should be lonely in the scope of $u$. Therefore, for a fixed choice of $u$, all the alternating paths starting from $u$ generated by Step-(2)-friendly pairs are different. Using the same argument as in \Cref{step1friendly} that $H$ does not contain any edges in $V_\lo\times V_\lo$, the total length of all such alternating paths is bounded by $4\Delta n_0$.
\end{proof}
\fi

\begin{claim}
	Under the condition that $\sum_{u\in X}\desc_\psi(u)\geq \frac{d|X|}{8}$, the total number of Step-(2)-friendly pairs is at least $d|X|/16$.
\end{claim}
\begin{proof}
	On the one hand, as $\sum_{u\in X}\desc_\psi(u)\geq \frac{d|X|}{8}$, the number of pairs $(u, v)$ that satisfies (i)(ii) in \Cref{step2friend} is at least $\frac{d|X|}{8}$. On the other hand, using \Cref{sum-alt-path}, the total number of alternating paths with length $\geq 2$ is at most $6m$ (each single path could be counted at most twice), and therefore by \Cref{param-ineq} the total number of alternating paths whose lengths are larger than $L$ is at most:
	$$\frac{6m}{L} <  d|X|/16$$
	Therefore, by a subtraction the number of Step-(2)-friendly pairs is at least $d|X|/16$.
\end{proof}

According to our algorithm, for any Step-(2)-friendly pair $(u, v)$, the probability $(u, v)$ is sampled and processed in Step (2) is at least $\frac{1}{2d|X|}$. Therefore, under the condition that $\sum_{u\in X}\desc_\psi(u)\geq \frac{d|X|}{8}$, Step (2) would succeed with probability at least $1/32$.

\paragraph{Success condition of Step \ref{alg-step3}} For a proof overview, please check \Cref{overview-step3}. Next, we assume that the success conditions of Step (1)(2) do not hold; that is $\sum_{u\in X}\lone_\psi(u) >  \frac{d|X|}{2}$ and $\sum_{u\in X}\desc_\psi(u)< \frac{d|X|}{8}$. For any $u\in T$ and any tree component $T\subseteq T_\psi(u)$, define $\br_\psi(T) = \sum_{b\in T, \br_T(b) \geq 2}\br_T(b)$, and let $\br_\psi(u) = \sum_{T\subseteq T_\psi(u)}\br_\psi(T)$.

\begin{claim}\label{tree}
	For any $u\in X$, $\br_\psi(u)< 4d$, and $\br_\psi(u) + 2\desc_\psi(u)\geq \lone_\psi(u)$.
\end{claim}
\begin{proof}
	Take any tree component $T\subseteq T_\psi(u)$, let $D\subseteq T$ be the set of all vertices such that for each $w\in D$, $\lst_\psi(u, \psi(u, w))$ is a singleton set. For the rest we only consider the case where $D\neq \emptyset$; otherwise we move on to other tree components of $T_\psi(u)$.
 
    Let $B\subseteq T$ be the set of all branching vertices in $T$. Let $F\subseteq D$ be the set of vertices which have strict descendants from $D$ in tree $T$. Therefore $D\setminus F$ is the set of vertices in $D$ that are on the lowest positions of $T$. 
	
	For the first half of the statement, consider the subtree of $T$ formed by vertices in $B\cup (D\setminus F)$. Then, we have:
	$$\frac{1}{2}\sum_{b\in B}\br_T(b) + |D\setminus F|\geq |B| + |D\setminus F| = 1+ \sum_{b\in B}\br_T(b)$$
	which leads to $\br_\psi(T) = \sum_{b\in B}\br_T(b) < 2|D\setminus F|\leq 2|D|$. Taking a summation over all tree component $T$, we have $\sum_{b\in T}\br_T(b)< 4d$.
	
	For the second half of the statement, consider three cases below.
	\begin{itemize}[leftmargin=*]
		\item If $T$ only contains one vertex $w$ such that $\lst_\psi(u, \psi(u, w))$ contains a lone neighbor $v$ of $u$, then the pair $(u, v)$ is counted in the sum $\desc_\psi(u)$, while $D = \{w\}$ contains only one vertex. 
		
		\item Otherwise, if $|D|\geq 2$ but all vertices in $D$ are on the same leaf-to-root tree path in $T$, then $|D\setminus F|\geq |D|-1\geq |D|/2$. By definition of $\desc_\psi(u)$, this tree $T$ would contribute $|D|-1\geq |D|/2$ to the total sum $\desc_\psi(u)$, and contribute $|D|$ to $\lone_\psi(u)$.
		
		\item Finally, if $|D\geq 2|$ and vertices in $D$ are not on the same leaf-to-root tree path in $D$, then because $|D| - |F|$ is the number of vertices in $D$ that are on lowest positions of $T$, this quantity should not exceed the total number of branchings which is $\br_\psi(T)$. In other words, $T$ contributes $|D| - |F|$ to the sum $\br_\psi(u)$, and $|D|$ to the sum $\lone_\psi(u)$.
	\end{itemize}
	Summing over all the three cases and all different tree components in $T_\psi(u)$, we have $\br_\psi(u) + 2\desc_\psi(u)\geq \lone_\psi(u)$. See \Cref{branch-leaf} for an illustration.
	
	\begin{figure}
		\centering
		\input{figs/branch-leaf}
		\caption{This example shows a tree component $T$ where $D\setminus F$ are red vertices, $F$ are orange vertices, and branching vertices are blue ones. So by definition, $\br_\psi(T) = 4 = 1+|D\setminus F|$.}\label{branch-leaf}
	\end{figure}
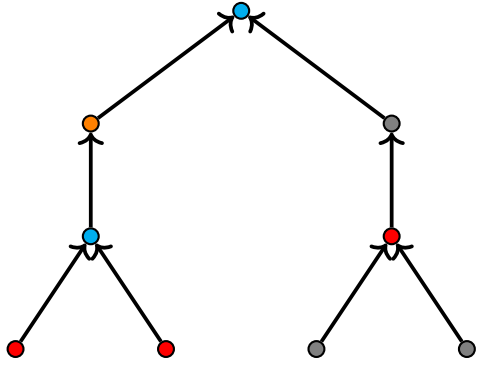
\end{proof}

By the above claim, under the assumption that $\sum_{u\in X}\lone_\psi(u) >  \frac{d|X|}{2}$ and $\sum_{u\in X}\desc_\psi(u)< \frac{d|X|}{8}$, we know that $\sum_{u\in X}\br_\psi(u)>\frac{d|X|}{4}$. Therefore, there exists an index $h\in [1, \log\Delta]$ such that: $$2^{h+1}\sum_{u\in X}|B_h(u)|>\sum_{u\in X}\sum_{b\in T\subseteq T_\psi(u), \br_T(b)\in [2^h, 2^{h+1})}\br_T(b)\geq \frac{d|X|}{4\log\Delta}$$
For each vertex $u\in X$, let $B_h(u)\subseteq T_\psi(u)$ be the set of branching vertices whose branch numbers in $T_\psi(u)$ are in the range $[2^h, 2^{h+1})$. Then, by \Cref{tree}, we have:
$$|B_h(u)|\leq \br_\psi(u) / 2^h < d / 2^{h-2}$$

\begin{definition}\label{def-cong}
	For any vertex $w$, define the \emph{congestion} $\cng(w)$ of $w$ to be the total number of vertices $u\in X$ such that $w$ is an active child of some branching vertex $b$ in a tree component $T\subseteq T_\psi(u)$ with $\br_T(b)\in [2^h, 2^{h+1})$. A vertex $w$ is called \emph{congested} if $\cng(w) \geq \tau$.
\end{definition}

Define $C = \{w\mid \cng(w) \geq 1\}$, so by definition of congestion we have: $$\sum_{w\in C}\cng(w) = \sum_{u\in X}\sum_{b\in T\subseteq T_\psi(u), \br_T(b)\in [2^h, 2^{h+1})}\br_T(b)\geq \frac{d|X|}{4\log\Delta}$$
Decompose this set as $C = C_\text{low}\cup C_\text{high}$ where $C_\text{high} = \{w\mid \cng(w) \geq \tau\}$. Next, we will show that a sufficient condition for the success probability of Step (3) is:
$$\sum_{w\in C_\text{low}}\cng(w)> \frac{2}{3}\sum_{w\in C}\cng(w)$$

\begin{definition}
	For a vertex $u\in X$ as well as a colored neighbor $w$ of $u$ such that $w\in C_\text{low}$ and the value of $\cng(w)$ is contributed by $u$. Let $R_{x, c_\psi(w)}(w)$ be the $(x, c_\psi(w))$-alternating path starting at $w$. We say that $(u, w)$ is \emph{Step-(3)-ready} if $R_{x, c_\psi(w)}(w)$ does not terminate at $u$.
\end{definition}

\begin{claim}
	Under the condition that $\sum_{w\in C_\text{low}}\cng(w)> \frac{2}{3}\sum_{w\in C}\cng(w)$, the total number of Step-(3)-ready pairs is at least $\frac{1}{6}\sum_{w\in C}\cng(w) \geq \frac{d|X|}{24\log\Delta}$.
\end{claim}
\begin{proof}
	Consider any pair $(u, w)$ such that $w$ is not congested. Let $b$ be the parent node of $w$ in $T_\psi(u)$. We say that $w$ is \emph{single} in $T_\psi(u)$ if all active children of $b$ other than $w$ are congested.
	
	By definition, we can charge any single vertex $w$ to all of its siblings vertices in $T_\psi(u)$, and thus the total number of pairs $(u, w)$ such that $w$ is single in $T_\psi(u)$ is at most:
	$$\sum_{w'\in C_\text{high}}\cng(w') < \frac{1}{3}\sum_{w'\in C}\cng(w')$$
	Consequently, the total number of pairs $(u, w)$ such that $w\in C_\text{low}$ is not single in $T_\psi(u)$ is at least:
	$$\sum_{w'\in C_\text{low}}\cng(w') -  \frac{1}{3}\sum_{w'\in C}\cng(w') > \frac{1}{3}\sum_{w'\in C}\cng(w') $$
	
	For each $u\in X$, for each branching vertex $b\in T\subseteq T_\psi(u)$ such that $\br_T(b)\in [2^h, 2^{h+1})$, let $w_1, w_2, \ldots, w_l\in C_\text{low}$ ($l\geq 2$) be all of its non-single vertices. Then, at most one of the $(x, \psi(u, b))$-alternating paths $R_{x, \psi(u, b)}(w_1), R_{x, \psi(u, b)}(w_2), \ldots, R_{x, \psi(u, b)}(w_l)$ can end at $u$; or in other words, at least $l-1 \geq l/2$ of the pairs $(u, w_1), (u, w_2), \ldots, (u, w_l)$ are Step-(3)-ready. Summing over all $u\in X$ we can prove that the total number of such pairs is at least $\frac{1}{6}\sum_{w'\in C}\cng(w')$.
\end{proof}

Finally, let us bound the total length of alternating paths.
\begin{claim}
	The total length of alternating paths $$\sum_{(u, w)\text{ is Step-(3)-ready}, |R_{x, c_\psi(w)}(w)|\geq 2}|R_{x, c_\psi(w)}(w)|$$ is at most $6m\tau$.
\end{claim}
\begin{proof}
	By definition of Step-(3)-readiness, for each such pair $(u, w)$ in the summation, we have $\cng(w)\leq \tau$. Therefore, the same alternating path length $|R_{x, c_\psi(w)}(w)|$ appears in the summation for at most $\tau$ times. Then, using \Cref{sum-alt-path}, we know that the total summation is at most $6m\tau$.
\end{proof}

\begin{claim}
	A Step-(3)-ready pair $(u, w)$ is called Step-(3)-friendly if $|R_{x, c_\psi(w)}(w)|\leq L$. Then, the total number of Step-(3)-friendly pairs is at least $\frac{d|X|}{48\log\Delta}$.
\end{claim}
\begin{proof}
	Since the total number of Step-(3)-ready pairs is at least $\frac{d|X|}{24\log\Delta}$, then by the previous claim and \Cref{param-ineq}, we know that the total number of Step-(3)-friendly pairs is at least:
	$$\frac{d|X|}{24\log\Delta} - \frac{6m\tau}{L} > \frac{d|X|}{48\log\Delta}$$
\end{proof}

According to the algorithm and using \Cref{tree}, for each Step-(3)-friendly pair $(u, w)$, it is selected by Step (3) with probability at least:
$$\frac{1}{|X|}\cdot \frac{1}{\sum_{b\in B_h(u)}\br_T(b)}> \frac{1}{4d|X|}$$

Under the condition that 
$$\sum_{w\in C_\text{low}}\cng(w)> \frac{2}{3}\sum_{w\in C}\cng(w) \geq \frac{d|X|}{6\log\Delta}$$ 
we know that with probability at least $\frac{1}{192\log^2\Delta}$ (the extra log factor is the inverse probability of guessing the right choice of $h$), the pair $(u, w)$ is Step-(3)-friendly, and therefore Step (3) would succeed in this case.

\paragraph{Success condition of Step \ref{alg-step4}} For a proof overview, please check \Cref{overview-step4}. Finally, let us assume that the success conditions of Step (1)(2)(3) do not hold. According to previous analysis, we can assume now $\sum_{w\in C_\text{high}}\cng(w)> \frac{1}{3}\sum_{w\in C}\cng(w)$. 

Let $C_h(u)\subseteq B_h(u)$ be the set of branching vertices one of whose active children is congested. 
\begin{claim}\label{step4friendly}
	$\sum_{u\in X}|C_h(u)|\geq \frac{1}{6}\sum_{u\in X}|B_h(u)| > \frac{d|X|}{3\cdot 2^{h+4}\log\Delta}$.
\end{claim}
\begin{proof}
	Since all branch numbers are in the range $[2^h, 2^{h+1})$, we have:
	$$\sum_{u\in X}|C_h(u)|\geq \frac{1}{2^{h+1}}\sum_{w\in C_\text{high}}\cng(w) > \frac{1}{6\cdot 2^{h}}\sum_{w\in C}\cng(w) \geq \frac{1}{6}\sum_{u\in X}|B_h(u)| >  \frac{d|X|}{48\cdot 2^{h}\log\Delta}$$
\end{proof}

Recall that in Step (4) the algorithms uniformly samples a vertex $b\in B_h(u)$ and try to compute the alternating path $S_{x, \psi(u, b)}(b)$. Let us group all alternating paths $S_{x, \psi(u, b)}(b)$ for $b\in \bigcup_{u\in X}C_h(u)$. For each fixed color $z\in \{1, 2, \ldots, \Delta+1\}$, let $\paths_{x, z}$ be the collection of all $\{x, z\}$-alternating paths in the set $\{S_{x, \psi(u, b)}(b)\mid b\in C_h(u), u\in X\}$.
\begin{claim}\label{step4total}
	For any $z\in \{1, 2, \ldots, \Delta+1\}$ such that $\paths_{x, z}\neq\emptyset$, we have $|\paths_{x, z}|\geq \tau$, and the total length of paths in $\paths_{x, z}$ is at most $2n$.
\end{claim}
\begin{proof}
	Take any path $S_{x, \psi(u, b)}(b)\in \paths_{x, z}$. By definition of $b$, there is a congested child of $w$ in $T_\psi(u)$. By definition of congestion, there exists at least $\tau-1$ different choices of $u'\in X$, such that $w$ is an active child of a branching vertex $b'\in B_h(u')$. Therefore, since $(w, b')$ is a directed edge in $T_\psi(u')$, we know that $\psi(u', b') = c_\psi(w) = z$. So, each alternating path $S_{x, \psi(u', b')}(b')$ is also in $\paths_{x, z}$.
	
	The second half of the statement can be proved using the property that each alternating path is counted at most twice, and so total length sum is at most $2n$.
\end{proof}

\begin{claim}
	For any $u\in X$ and $b\in C_h(u)$, we say that $(u, b)$ is Step-(4)-friendly, if $|S_{x, \psi(u, b)}(b)|\leq L$. Then, the total number of Step-(4)-friendly pairs is at least $\frac{d|X|}{100\cdot 2^{h}\log\Delta}$.
\end{claim}
\begin{proof}
	By \Cref{step4total}, for each $z\in \{1, 2, \ldots, \Delta+1\}$, the total number of paths in $\paths_{x, z}$ with length at most $L$ is at least:
	$$|\paths_{x, z}| - \frac{2n}{L} \geq \frac{1}{2}|\paths_{x, z}|$$
	as $\frac{1}{2}|\paths_{x, z}|\geq \frac{\tau}{2} \geq \frac{2n}{L}$. Therefore, at least half of the paths in $\paths_{x, z}$ are coming from Step-(4)-friendly pairs. So, the total number of Step-(4)-friendly pairs should be at least $\frac{1}{2}\sum_{u\in X}|C_h(u)| > \frac{d|X|}{100\cdot 2^h\log\Delta}$.
\end{proof}

According to the algorithm, for any Step-(4)-friendly pair $(u, b)$, Step (4) selects this pair with probability at least $$\frac{1}{|X|\log\Delta}\cdot \frac{1}{|B_h(u)|}\geq \frac{2^{h-2}}{d|X|\log\Delta}$$
Therefore, Step (4) succeeds with at least constant probability at least $\frac{1}{100\log^2\Delta}$.

This concludes the proof of 
\Cref{succ}.

\subsection{Proof of \Cref{faster-vizing}}\label{large-Delta}

Finally, let us finish our main theorem. By \Cref{faster-vizing-small}, the main theorem holds whenever $\Delta \leq n^{2/3}$, so we are left with the case that $\Delta > n^{2/3}$.

The algorithm is by a recursive application of Eulerian partition provided by \Cref{euler}. Basically, on input tuple $\brac{G_0 = (V, E_0), \Delta_0}$, if $\Delta_0 \leq n^{2/3}$, then simply run the algorithm from \Cref{faster-vizing-small} to find a $(\Delta_0+1)$-edge coloring in time $\tilde{O}(mn^{1/3})$. Otherwise, apply \Cref{euler} on $G_0$, which produces a partition $G_0 = G_1\cup G_2$ in linear time $O(|E_0|)$ into two subgraphs $G_1$ and $G_2$ with maximum degrees $\floor{\Delta_0/2}\leq \Delta_1, \Delta_2\leq \ceil{\Delta_0/2}$, respectively. After that, recursively compute two edge colorings with disjoint color palettes $\psi_1, \psi_2$ on tuples $\brac{G_1, \Delta_1}$ and $\brac{G_2, \Delta_2}$ respectively.

Since $\psi_1, \psi_2$ are using disjoint color palettes, the total number of different colors in $\psi_1, \psi_2$ is at most $\Delta_0+3$. To remove the extra colors, pick the two least popular colors under $\psi_1, \psi_2$ and uncolor all the edges with these two colors. Now let $\psi$ be the union of $\psi_1, \psi_2$, which is now a partial coloring of $E_0$ with at most $2|E_0|/\Delta_0$ uncolored edges. Then, apply the standard Vizing's procedure to extend $\psi$ to these edges, which takes total time $\tilde{O}(|E_0|n / \Delta_0) = \tilde{O}(|E_0|n^{1/3})$, which concludes the proof.

\newpage

%% file: figs/step2-1.tex
\begin{tikzpicture}[thick,scale=1]
	\draw (0, 2) node(1)[circle, draw, fill=black!50,
	inner sep=0pt, minimum width=6pt, label = $u$] {};
	
	\draw (3, 0) node(2)[circle, draw, fill=black!50,
	inner sep=0pt, minimum width=6pt, label = {-45: {$w$}}] {};
	
	\draw (1, 0) node(3)[circle, draw, fill=black!50,
	inner sep=0pt, minimum width=6pt, label = {-90: {$w'$}}] {};
	
	\draw (-1, 0) node(4)[circle, draw, fill=black!50,
	inner sep=0pt, minimum width=6pt, label = {-135: {$v$}}] {};
	\draw (-3, 0) node(5)[circle, draw, fill=black!50,
	inner sep=0pt, minimum width=6pt, label = {-90: {$v'$}}] {};
	
	\draw (2, -2) node(6)[circle, draw, fill=black!50,
	inner sep=0pt, minimum width=6pt] {};
	\draw (0, -2) node(7)[circle, draw, fill=black!50,
	inner sep=0pt, minimum width=6pt] {};
	
	\draw [line width = 0.5mm, color=orange] (1) to (2);	
	\draw [line width = 0.5mm, color=red] (1) to (3);
	\draw [->, line width = 0.5mm] (3) to (2);
	\draw [line width = 0.5mm, dashed] (1) to (4);
	\draw [line width = 0.5mm, dashed] (1) to (5);
	\draw [line width = 0.5mm, color=cyan] (2) to (6);
	\draw [line width = 0.5mm, color=orange] (6) to (7);
	\draw [line width = 0.5mm, color=cyan] (7) to (4);
	
	\draw (1, -3) node[black, label={$Q_{x, y}$}]{};
	
	\draw [->, line width = 2mm] (4, 0) to (6, 0);
	%========================
	
	\draw (10, 2) node(8)[circle, draw, fill=black!50,
	inner sep=0pt, minimum width=6pt, label = $u$] {};
	
	\draw (13, 0) node(9)[circle, draw, fill=black!50,
	inner sep=0pt, minimum width=6pt, label = {-45: {$w$}}] {};
	
	\draw (11, 0) node(10)[circle, draw, fill=black!50,
	inner sep=0pt, minimum width=6pt, label = {-90: {$w'$}}] {};
	
	\draw (9, 0) node(11)[circle, draw, fill=black!50,
	inner sep=0pt, minimum width=6pt, label = {-135: {$v$}}] {};
	\draw (7, 0) node(12)[circle, draw, fill=black!50,
	inner sep=0pt, minimum width=6pt, label = {-90: {$v'$}}] {};
	
	\draw (12, -2) node(13)[circle, draw, fill=black!50,
	inner sep=0pt, minimum width=6pt] {};
	\draw (10, -2) node(14)[circle, draw, fill=black!50,
	inner sep=0pt, minimum width=6pt] {};
	
	\draw [line width = 0.5mm, color=cyan] (8) to (9);	
	\draw [line width = 0.5mm, color=orange] (8) to (10);
	%\draw [->, line width = 0.5mm] (10) to (9);
	\draw [line width = 0.5mm, dashed] (8) to (11);
	\draw [line width = 0.5mm, color=red] (8) to (12);
	\draw [line width = 0.5mm, color=orange] (9) to (13);
	\draw [line width = 0.5mm, color=cyan] (13) to (14);
	\draw [line width = 0.5mm, color=orange] (14) to (11);
	
	\draw (11, -3) node[black, label={$Q_{x, y}$}]{};
\end{tikzpicture}

%% file: figs/step2-2.tex
\begin{tikzpicture}[thick,scale=1]
	\draw (0, 2) node(1)[circle, draw, fill=black!50,
	inner sep=0pt, minimum width=6pt, label = $u$] {};
	
	\draw (3, 0) node(2)[circle, draw, fill=black!50,
	inner sep=0pt, minimum width=6pt, label = {-45: {$w$}}] {};
	
	\draw (1, 0) node(3)[circle, draw, fill=black!50,
	inner sep=0pt, minimum width=6pt, label = {180: {$w'$}}] {};
	
	\draw (-1, 0) node(4)[circle, draw, fill=black!50,
	inner sep=0pt, minimum width=6pt, label = {-135: {$v$}}] {};
	\draw (-3, 0) node(5)[circle, draw, fill=black!50,
	inner sep=0pt, minimum width=6pt, label = {-90: {$v'$}}] {};
	
	\draw (2, -2) node(6)[circle, draw, fill=black!50,
	inner sep=0pt, minimum width=6pt] {};
	\draw (0, -2) node(7)[circle, draw, fill=black!50,
	inner sep=0pt, minimum width=6pt] {};
	
	\draw [line width = 0.5mm, color=orange] (1) to (2);	
	\draw [line width = 0.5mm, color=red] (1) to (3);
	\draw [->, line width = 0.5mm] (3) to (2);
	\draw [line width = 0.5mm, dashed] (1) to (4);
	\draw [line width = 0.5mm, dashed] (1) to (5);
	\draw [line width = 0.5mm, color=cyan] (2) to (6);
	\draw [line width = 0.5mm, color=orange] (6) to (7);
	\draw [line width = 0.5mm, color=cyan] (7) to (3);
	
	\draw (1, -3) node[black, label={$Q_{x, y}$}]{};
	
	\draw [->, line width = 2mm] (4, 0) to (6, 0);
	%========================
	
	\draw (10, 2) node(8)[circle, draw, fill=black!50,
	inner sep=0pt, minimum width=6pt, label = $u$] {};
	
	\draw (13, 0) node(9)[circle, draw, fill=black!50,
	inner sep=0pt, minimum width=6pt, label = {-45: {$w$}}] {};
	
	\draw (11, 0) node(10)[circle, draw, fill=black!50,
	inner sep=0pt, minimum width=6pt, label = {180: {$w'$}}] {};
	
	\draw (9, 0) node(11)[circle, draw, fill=black!50,
	inner sep=0pt, minimum width=6pt, label = {-135: {$v$}}] {};
	\draw (7, 0) node(12)[circle, draw, fill=black!50,
	inner sep=0pt, minimum width=6pt, label = {-90: {$v'$}}] {};
	
	\draw (12, -2) node(13)[circle, draw, fill=black!50,
	inner sep=0pt, minimum width=6pt] {};
	\draw (10, -2) node(14)[circle, draw, fill=black!50,
	inner sep=0pt, minimum width=6pt] {};
	
	\draw [line width = 0.5mm, color=cyan] (8) to (9);	
	\draw [line width = 0.5mm, color=red] (8) to (10);
	%\draw [->, line width = 0.5mm] (10) to (9);
	\draw [line width = 0.5mm, color=orange] (8) to (11);
	\draw [line width = 0.5mm, dashed] (8) to (12);
	\draw [line width = 0.5mm, color=orange] (9) to (13);
	\draw [line width = 0.5mm, color=cyan] (13) to (14);
	\draw [line width = 0.5mm, color=orange] (14) to (10);
	
	\draw (11, -3) node[black, label={$Q_{x, y}$}]{};
\end{tikzpicture}

%% file: figs/step3.tex
\begin{tikzpicture}[thick,scale=0.9]
	\draw (0, 2) node(1)[circle, draw, fill=black!50,
	inner sep=0pt, minimum width=6pt, label = $u$] {};
	
	\draw (3, 0) node(2)[circle, draw, fill=black!50,
	inner sep=0pt, minimum width=6pt, label = {0: {$b$}}] {};
	
	\draw (1, 0) node(3)[circle, draw, fill=black!50,
	inner sep=0pt, minimum width=6pt, label = {-90: {$c$}}] {};
	
	\draw (-1, 0) node(4)[circle, draw, fill=black!50,
	inner sep=0pt, minimum width=6pt, label = {-135: {$s$}}] {};
	\draw (-3, 0) node(5)[circle, draw, fill=black!50,
	inner sep=0pt, minimum width=6pt, label = {-90: {$v'$}}] {};
	
	\draw (2, -2) node(6)[circle, draw, fill=black!50,
	inner sep=0pt, minimum width=6pt] {};
	\draw (0, -2) node(7)[circle, draw, fill=black!50,
	inner sep=0pt, minimum width=6pt] {};
	\draw (-2, -2) node(8)[circle, draw, fill=black!50,
	inner sep=0pt, minimum width=6pt] {};
	\draw (-3, -4) node(9)[circle, draw, fill=black!50,
	inner sep=0pt, minimum width=6pt] {};
	
	\draw [line width = 0.5mm, color=orange] (1) to (2);	
	\draw [line width = 0.5mm, color=red] (1) to (3);
	\draw [->, line width = 0.5mm] (3) to (2);
	\draw [->, line width = 0.5mm] (4) to (3);
	\draw [line width = 0.5mm, color=purple] (1) to (4);
	\draw [line width = 0.5mm, dashed] (1) to (5);
	\draw [->, line width = 0.5mm] (6) to (2);
	\draw [line width = 0.5mm, color=cyan] (7) to (3);
	\draw [line width = 0.5mm, color=orange] (7) to (8);
	\draw [line width = 0.5mm, color=cyan] (8) to (9);
	
	\draw (-2, -4) node[black, label={$R_{x, z}$}]{};
	
	\draw [->, line width = 2mm] (4, 0) to (6, 0);
	%========================
	
	\draw (10, 2) node(10)[circle, draw, fill=black!50,
	inner sep=0pt, minimum width=6pt, label = $u$] {};
	
	\draw (13, 0) node(11)[circle, draw, fill=black!50,
	inner sep=0pt, minimum width=6pt, label = {0: {$b$}}] {};
	
	\draw (11, 0) node(12)[circle, draw, fill=black!50,
	inner sep=0pt, minimum width=6pt, label = {-90: {$c$}}] {};
	
	\draw (9, 0) node(13)[circle, draw, fill=black!50,
	inner sep=0pt, minimum width=6pt, label = {-135: {$s$}}] {};
	\draw (7, 0) node(14)[circle, draw, fill=black!50,
	inner sep=0pt, minimum width=6pt, label = {-90: {$v'$}}] {};
	
	\draw (12, -2) node(15)[circle, draw, fill=black!50,
	inner sep=0pt, minimum width=6pt] {};
	\draw (10, -2) node(16)[circle, draw, fill=black!50,
	inner sep=0pt, minimum width=6pt] {};
	\draw (8, -2) node(17)[circle, draw, fill=black!50,
	inner sep=0pt, minimum width=6pt] {};
	\draw (7, -4) node(18)[circle, draw, fill=black!50,
	inner sep=0pt, minimum width=6pt] {};
	
	\draw [line width = 0.5mm, color=orange] (10) to (11);	
	\draw [line width = 0.5mm, color=cyan] (10) to (12);
	%\draw [->, line width = 0.5mm] (12) to (11);
	%\draw [->, line width = 0.5mm] (13) to (12);
	\draw [line width = 0.5mm, color=red] (10) to (13);
	\draw [line width = 0.5mm, color=purple] (10) to (14);
	\draw [->, line width = 0.5mm] (15) to (11);
	\draw [line width = 0.5mm, color=orange] (16) to (12);
	\draw [line width = 0.5mm, color=cyan] (16) to (17);
	\draw [line width = 0.5mm, color=orange] (17) to (18);
	
	\draw [->, line width = 0.5mm] (13) to (14);
	\draw [->, line width = 0.5mm] (12) to (13);
	
	\draw (8, -4) node[black, label={$R_{x, z}$}]{};
\end{tikzpicture}

%% file: figs/step4.tex
\begin{tikzpicture}[thick,scale=0.9]
	\draw (0, 2) node(1)[circle, draw, fill=black!50,
	inner sep=0pt, minimum width=6pt, label = $u$] {};
	
	\draw (3, 0) node(2)[circle, draw, fill=black!50,
	inner sep=0pt, minimum width=6pt, label = {0: {$b$}}] {};
	
	\draw (1, 0) node(3)[circle, draw, fill=black!50,
	inner sep=0pt, minimum width=6pt, label = {-90: {$c_1$}}] {};
	
	\draw (-1, 0) node(4)[circle, draw, fill=black!50,
	inner sep=0pt, minimum width=6pt, label = {-90: {$t$}}] {};
	\draw (-3, 0) node(5)[circle, draw, fill=black!50,
	inner sep=0pt, minimum width=6pt, label = {-90: {$v'$}}] {};
	
	\draw (1, -2) node(6)[circle, draw, fill=black!50,
	inner sep=0pt, minimum width=6pt, label = {-90: {$c_2$}}] {};
	\draw (3, -2) node(7)[circle, draw, fill=black!50,
	inner sep=0pt, minimum width=6pt] {};
	\draw (3, -4) node(8)[circle, draw, fill=black!50,
	inner sep=0pt, minimum width=6pt] {};
	\draw (1, -4) node(9)[circle, draw, fill=black!50,
	inner sep=0pt, minimum width=6pt] {};
	
	\draw [line width = 0.5mm, color=orange] (1) to (2);	
	\draw [line width = 0.5mm, color=red] (1) to (3);
	\draw [->, line width = 0.5mm] (3) to (2);
	\draw [->, line width = 0.5mm] (4) to (3);
	\draw [line width = 0.5mm, color=purple] (1) to (4);
	\draw [line width = 0.5mm, dashed] (1) to (5);
	\draw [->, line width = 0.5mm] (6) to (2);
	\draw [line width = 0.5mm, color=cyan] (7) to (2);
	\draw [line width = 0.5mm, color=orange] (7) to (8);
	\draw [line width = 0.5mm, color=cyan] (8) to (9);
	
	\draw (2, -4) node[black, label={$S_{x, z}$}]{};
	
	\draw [->, line width = 2mm] (4, 0) to (6, 0);
	%========================
	
	\draw (10, 2) node(10)[circle, draw, fill=black!50,
	inner sep=0pt, minimum width=6pt, label = $u$] {};
	
	\draw (13, 0) node(11)[circle, draw, fill=black!50,
	inner sep=0pt, minimum width=6pt, label = {0: {$b$}}] {};
	
	\draw (11, 0) node(12)[circle, draw, fill=black!50,
	inner sep=0pt, minimum width=6pt, label = {-135: {$c_1$}}] {};
	
	\draw (9, 0) node(13)[circle, draw, fill=black!50,
	inner sep=0pt, minimum width=6pt, label = {-90: {$t$}}] {};
	\draw (7, 0) node(14)[circle, draw, fill=black!50,
	inner sep=0pt, minimum width=6pt, label = {-90: {$v'$}}] {};
	
	\draw (11, -2) node(15)[circle, draw, fill=black!50,
	inner sep=0pt, minimum width=6pt, label = {-90: {$c_2$}}] {};
	\draw (13, -2) node(16)[circle, draw, fill=black!50,
	inner sep=0pt, minimum width=6pt] {};
	\draw (13, -4) node(17)[circle, draw, fill=black!50,
	inner sep=0pt, minimum width=6pt] {};
	\draw (11, -4) node(18)[circle, draw, fill=black!50,
	inner sep=0pt, minimum width=6pt] {};
	
	\draw [line width = 0.5mm, color=cyan] (10) to (11);	
	\draw [line width = 0.5mm, color=orange] (10) to (12);
	\draw [->, line width = 0.5mm] (11) to (12);
	\draw [->, line width = 0.5mm] (12) to (13);
	\draw [line width = 0.5mm, color=red] (10) to (13);
	\draw [line width = 0.5mm, color=purple] (10) to (14);
	\draw [->, line width = 0.5mm] (15) to (12);
	\draw [line width = 0.5mm, color=orange] (16) to (11);
	\draw [line width = 0.5mm, color=cyan] (16) to (17);
	\draw [line width = 0.5mm, color=orange] (17) to (18);
	
	\draw [->, line width = 0.5mm] (13) to (14);
	
	\draw (12, -4) node[black, label={$S_{x, z}$}]{};
\end{tikzpicture}

%% file: figs/branch-leaf.tex
\begin{tikzpicture}[thick,scale=1]

	%\draw (-2, 0) node(1)[circle, draw, fill=black!50, inner sep=0pt, minimum width=6pt, label = {180: {$v_{i, j}$}}] {};
	
	\draw (-3, 0) node(1)[circle, draw, fill=red, inner sep=0pt, minimum width=6pt] {};
	\draw (-1, 0) node(2)[circle, draw, fill=red, inner sep=0pt, minimum width=6pt] {};
	\draw (1, 0) node(3)[circle, draw, fill=black!50, inner sep=0pt, minimum width=6pt] {};
	\draw (3, 0) node(4)[circle, draw, fill=black!50, inner sep=0pt, minimum width=6pt] {};
	
	\draw (-2, 1.5) node(5)[circle, draw, fill=cyan, inner sep=0pt, minimum width=6pt] {};
	\draw (2, 1.5) node(6)[circle, draw, fill=red, inner sep=0pt, minimum width=6pt] {};
	
	\draw (-2, 3) node(7)[circle, draw, fill=orange, inner sep=0pt, minimum width=6pt] {};
	\draw (2, 3) node(8)[circle, draw, fill=black!50, inner sep=0pt, minimum width=6pt] {};
	
	\draw (0, 4.5) node(9)[circle, draw, fill=cyan, inner sep=0pt, minimum width=6pt] {};
	
	%\draw [line width = 0.5mm, color=orange] (5) to (4);
	
	\draw [->, line width = 0.5mm] (1) to (5);
	\draw [->, line width = 0.5mm] (2) to (5);
	\draw [->, line width = 0.5mm] (3) to (6);
	\draw [->, line width = 0.5mm] (4) to (6);
	\draw [->, line width = 0.5mm] (5) to (7);
	\draw [->, line width = 0.5mm] (6) to (8);
	\draw [->, line width = 0.5mm] (7) to (9);
	\draw [->, line width = 0.5mm] (8) to (9);

\end{tikzpicture}